\documentclass{article}
\usepackage{fullpage,diagbox}
\usepackage{mathtools, appendix}

\mathtoolsset{showonlyrefs}
\newcommand{\clr}{\color{black}}

\usepackage[ruled,linesnumbered, vlined, noend]{algorithm2e}

\usepackage{amsthm, amsfonts, amssymb, amsmath,enumitem, natbib, bbm,mathabx,tcolorbox}
\newtheorem{thm}{Theorem}

\newtheorem{cor}{Corollary}
\newtheorem{prop}{Proposition}
\newtheorem{defi}{Definition}
\newcounter{myalgctr}
\RequirePackage[colorlinks=true, pdfstartview=FitV,
linkcolor=blue, citecolor=blue, urlcolor=blue,hypertexnames=false]{hyperref}
\numberwithin{myalgctr}{section}
\DeclareMathOperator*{\argmin}{arg\,min}

\usepackage{graphicx} 
\usepackage{color, float,authblk}
\title{\bf Exchangeability, Conformal Prediction,\\ and Rank Tests}
\author{Arun Kumar Kuchibhotla}
\affil{\texttt{arunku@cmu.edu}}
\affil{Department of Statistics \& Data Science, Carnegie Mellon University}
\date{}
\begin{document}

\maketitle

\begin{abstract}
Conformal prediction has been a very popular method of distribution-free predictive inference in recent years in machine learning and statistics. Its popularity stems from the fact that it works as a wrapper around any prediction algorithm such as neural networks or random forests. Exchangeability is at the core of the validity of conformal prediction. The concept of exchangeability is also at the core of rank tests widely known in nonparametric statistics. In this paper, we review the concept of exchangeability and discuss the implications for conformal prediction and rank tests. {\clr We provide a low-level introduction to these topics, and discuss the similarities between conformal prediction and rank tests.} 
\end{abstract}
\section{Introduction}\label{sec:intro}
Exchangeability of random variables is one of the fundamental concepts in statistics, probably right next to the concept of independent and identically distributed (i.i.d.) random variables. Although these two concepts are very closely related, the fact that exchangeability allows for a specific type of dependence between the random variables leads to numerous implications/applications of this concept. One of the most important implications of exchangeability is that the indexing of random variables is immaterial. In technical words, this means that the ranks of real-valued exchangeable random variables are uniform over the set of all permutations. Just this one implication has led to the development of two very different fields in statistics and machine learning, namely, non-parametric rank tests and conformal prediction. 

{\clr Conformal prediction fills an important gap in machine learning (prediction) algorithms and forecasting. For example, in the context of regression, classical algorithms only provide the point prediction for the response without any uncertainty quantification. Conformal prediction intervals centered at the point prediction algorithm act as such uncertainty measure. Further, classical prediction intervals such as those in linear regression are based on well-specified linear model assumptions. Conformal methods yield prediction regions without any such distributional assumptions.}

{\clr On the other hand, rank tests concern the statistical problem of testing hypotheses. Most commonly used hypothesis testing procedures asymptotically control the type I error and depend on certain distributional assumptions so as to ensure ``good'' asymptotic properties of the test statistic. Rank tests, when available, are finite sample valid and are distribution-free.}

The main purpose of this article is to define exchangeability, discuss its implications, and then exposit the uses of this concept for conformal prediction and rank tests. 
{\clr  
Both conformal prediction and rank tests make significant use of exchangeability to yield finite sample guarantees for prediction regions and type I error control, respectively, without any distributional assumptions. Conformal prediction regions for data in arbitrary dimensions has been well discussed in the literature, while rank tests for arbitrary dimensions is not as widely discussed. This is not to say that rank tests for multivariate or high-dimensional cases are unknown; see~\cite{friedman2003multivariate},~\cite{vayatis2009auc} for some works.
}
The popularity of conformal prediction stems from the fact that it can be wrapped around any arbitrary algorithm that provides point predictions and leads to a finite sample valid prediction regions for future observations. 

In this article, we show how non-parametric rank tests (usually defined for real-valued cases) can also be thought of as wrappers and be applied to arbitrary spaces. {\clr The idea of (data independent) dimension reduction for rank tests in arbitrary spaces trivially leads to a finite sample type I error control as mentioned in~\cite{matthews1996nonparametric} and~\cite{lheritier2015nonparametric}. Note that this also includes the case of transformations based on sample splitting as noted in~\cite{friedman2003multivariate} and~\cite{vayatis2009auc}. In this article, we show that the dimension reduction algorithm need not be independent of the data.}

Conformal prediction pioneered by~\cite{vovk2005algorithmic} was introduced to the statistics community by~\cite{MR3174619} and further explored in several works~\citep{lei2014distribution,lei2018distribution,chernozhukov2018exact,romano2019conformalized,barber2019limits,barber2019predictive,chernozhukov2019distributional,kuchibhotla2019nested}, among others. For a general overview of the topic, we refer the reader to~\cite{balasubramanian2014conformal}. The discussion in all of these papers starts with a ``conformity'' score. In this article, we do not formally define a ``conformity'' score but show the application of exchangeability and it is done also because similar thinking helps when we discuss rank tests.

The organization of the article is as follows. In Section~\ref{SEC:EXCHANGEABILITY-IMPLICATIONS}, we introduce the concept of exchangeability and discuss its implications for ranks of real-valued random variables. Exchangeability is a very intuitive concept that can make it hard to verify rigorously, in some cases. For this reason, in Section~\ref{subsec:transformations}, we discuss the issue of preserving exchangeability via transformations. In Section~\ref{SEC:CONFORMAL-PREDICTION}, we discuss the applications of the implication of exchangeability for the construction of distribution-free finite sample valid prediction regions. In Section~\ref{SEC:RANK-TESTS}, we discuss the applications of the implication of exchangeability for the construction of distribution-free finite-sample valid rank tests for testing equality of distributions as well as testing independence of two random variables. In both these tests of hypotheses, we allow the random variables to take values in an arbitrary space, thus showing the full strength of exchangeability for this application. In Section~\ref{sec:summary}, we summarize the article and discuss a few open questions.     

{\clr Most of the results in the article are either known or standard. All the results follow from the definition of exchangeability.}

\paragraph{Notation.} We use the following notation throughout the article. The notation $\overset{d}{=}$ represents the equality in distribution of two random variables. We abbreviate the set $\{1,2,\ldots,n\}$ by $[n]$ for any $n\ge1$. We write i.i.d. for independent and identically distributed.  
\section{Exchangeability and Implications}\label{SEC:EXCHANGEABILITY-IMPLICATIONS}
\subsection{Definition of Exchangeability}
Random variables $W_1, \ldots, W_n$ for $n\ge1$ are said to be exchangeable if
\begin{equation}\label{eq:exchangeability}
(W_1, \ldots, W_n) ~\overset{d}{=}~ (W_{\pi(1)}, \ldots, W_{\pi(n)}),
\end{equation}
for any permutation $\pi:[n]\rightarrow[n]$. Intuitively, exchangeability means that the index of the random variables is immaterial. If $W_1, \ldots, W_n$ are real-valued random variables, then the definition~\eqref{eq:exchangeability} is equivalent to the condition that $(W_1, \ldots, W_n)$ has the same (joint) cumulative distribution function as that of $(W_{\pi(1)}, \ldots, W_{\pi(n)})$, that is, for any $a_1, \ldots, a_n\in\mathbb{R}$, and any permutation $\pi:[n]\to[n]$,
\begin{equation}\label{eq:cumulative-distribution-function}
\mathbb{P}(W_1 \le a_1,\,\ldots,\,W_n \le a_n) ~=~ \mathbb{P}(W_{\pi(1)} \le a_1,\,\ldots,\, W_{\pi(n)} \le a_n).
\end{equation}
{\clr Here $\mathbb{P}(\cdot)$ represents the probability of the event with respect to the probability measure of $(W_1, \ldots, W_n)$.}
If $(W_1, \ldots, W_n)$ has a density $p(\cdot, \ldots, \cdot)$ with respect to the Lebesgue measure, then this condition is further equivalent to
\begin{equation}\label{eq:probability-density-function}
p(a_1,\,\ldots,\,a_n) ~=~ p(a_{\pi(1)},\,\ldots,\,a_{\pi(n)}),
\end{equation}
for any permutation $\pi:[n]\to[n]$ and any $a_1, \ldots, a_n\in\mathbb{R}$.
For random variable in an arbitrary measurable space $\mathcal{X}$, definition~\eqref{eq:exchangeability} is equivalent to
\begin{equation}\label{eq:general-exchangeability}
\mathbb{P}(W_1\in A_1,\,\ldots,\,W_n\in A_n) ~=~ \mathbb{P}(W_{\pi(1)} \in A_1,\,\ldots,\,W_{\pi(n)}\in A_{\pi(n)}),
\end{equation}
for any permutation $\pi:[n]\to[n]$ and any Borel measurable sets $A_1, \ldots, A_n$. A simple consequence of definition~\eqref{eq:general-exchangeability} is that exchangeable random variables must be identically distributed. To see this, fix a $j\in[n]$ and take the $\pi:[n]\to[n]$ such that $\pi(1) = j$. Choosing $A_2 = A_2 = \ldots = A_n = \mathcal{X}$ in~\eqref{eq:general-exchangeability} yields $\mathbb{P}(W_1 \in A_1) = \mathbb{P}(W_j \in A_1)$ and because $j\in[n]$ is arbitrary, the result follows. Hence, \emph{identical distributions is a necessary (but not a sufficient) condition for exchangeability.}

Further, it is not hard to verify using~\eqref{eq:general-exchangeability} that if $W_1, \ldots, W_n$ are independent and identically distributed (i.i.d.), then they are exchangeable.
\subsection{Examples and Counter-examples}\label{subsec:examples}
In the following, we provide a few examples of exchangeable random variables.
\begin{enumerate}
  \item Suppose $W_1, W_2$ have the joint distribution
  \begin{equation}\label{eq:bivariate-normal}
  \begin{pmatrix}W_1\\W_2\end{pmatrix} ~\sim~ N\left(\begin{pmatrix}0\\0\end{pmatrix},\,\begin{pmatrix}1 & \rho\\\rho & 1\end{pmatrix}\right).  
  \end{equation}
  For any $\rho \in [-1, 1]$, exchangeability of $W_1$ and $W_2$ can be verified readily using~\eqref{eq:probability-density-function}. Suppose the mean of the bivariate normal distribution in~\eqref{eq:bivariate-normal} is changed to $(\mu_1, \mu_2)^{\top}$. In this case, $W_1$ and $W_2$ are not exchangeable unless $\mu_1 \neq \mu_2$. This follows by noting that $W_1$ and $W_2$ do not have the same marginal distribution if $\mu_1 \neq \mu_2$. 
  \item Suppose $W_i = X_i + Z$ for i.i.d. random variables $X_i$ and another random variable $Z$ independent of $X_1, \ldots, X_n$. Then $W_1, \ldots, W_n$ are exchangeable. To prove this, note that
  \begin{align*}
  \mathbb{P}(W_1 \le w_1, \ldots, W_n \le w_n) &= \mathbb{E}[\mathbb{P}(W_1 \le w_1, \ldots, W_n \le w_n\big|Z)]\\
  &= \mathbb{E}[\mathbb{P}(X_1 \le w_1 - Z, \ldots, X_n \le w_n - Z\big|Z)]\\
  &\overset{(a)}{=} \mathbb{E}[\mathbb{P}(X_1 \le w_1 - Z\big|Z)\cdots\mathbb{P}(X_n \le w_n - Z\big|Z)]\\
  &\overset{(b)}{=} \mathbb{E}[P_X(w_1 - Z)\cdots P_X(w_n - Z)]\\
  &= \mathbb{P}(W_{\pi(1)} \le w_1, \ldots, W_{\pi(n)} \le w_n).
  \end{align*}
  Here (a) follows from the assumption that $X_1, \ldots, X_n$ are independent and (b) follows from the assumption that $X_i$'s are identically distributed with distribution function $P_X(\cdot)$. Finally, the last equality follows by retracing the steps with a permutation. {\clr Observe that the random variables $W_1, \ldots, W_n$ are not independent and their dependence stems from the common random variable $Z$. This shows that exchangeability in general does not imply independence.}
  \item Suppose $W_i = f(X_i, Z)$ for a function $f$, i.i.d. random variables $X_1, \ldots, X_n$ and another random variable $Z$ independent of $X_1, \ldots, X_n$. Then $W_1, \ldots, W_n$ are exchangeable. The proof is almost verbatim as in the previous case.
\end{enumerate}
\subsection{De Finetti's Theorem}
A commonality of the last two examples in Section~\ref{subsec:examples} is that there exists a random variable $Z$ conditional on which $W_1, \ldots, W_n$ are i.i.d.. This represents one of the most general ways of constructing exchangeable random variables. One of the most important results in Bayesian statistics states that if $n = \infty$, then there is no other way of constructing exchangeable random variables. Formally, {\clr we have the following De Finetti's representation theorem for infinite sequence of exchangeable random variables. The following statement is taken from~\citet[Theorem 1.49]{schervish2012theory}.
\begin{thm}[De Finetti's Representation Theorem]
Let $(S,\mathcal{A}, \mu)$ be a probability space, and let $(\mathcal{W}, \mathcal{B})$ be a Borel space. For each $n$, let $W_n: S\to\mathcal{X}$ be measurable. The sequence $\{W_i\}_{i=1}^{\infty}$ is exchangeable if and only if there is a random probability measure $\mathbf{P}$ on $(\mathcal{W}, \mathcal{B})$ such that, conditional on $\mathbf{P} = P$, $\{W_i\}_{i=1}^{\infty}$ are independent and identically distributed with distribution $P$. Furthermore, if the sequence is exchangeable, then the distribution of $\mathbf{P}$ is unique, and $n^{-1}\sum_{i=1}^n \mathbbm{1}\{W_i\in B\}$ converges to $\mathbf{P}(B)$ almost surely for each $B\in\mathcal{B}$.
\end{thm}
\cite{de1929funzione} proved the result only for random variables $W_i$ taking values in $\{0, 1\}$. It was extended to arbitrary compact Hausdorff spaces by~\cite{hewitt1955symmetric}. Interested readers can refer to~\cite{ressel1985finetti} and~\cite{aldous1985exchangeability} for a review and a detailed discussion of probabilistic aspects of exchangeability.
The hypothesis that there are infinite number of elements in the sequence is crucial and it can be shown that the result is false for a finite sequence of exchangeable random variables, in general. See~\citet[Chapter 1, Section 1.2]{schervish2012theory} for a counterexample. Also, see Theorem 1.48 of~\cite{schervish2012theory} for a representation theorem for a finite sequence of exchangeable random variables.}
\subsection{Implication for Ranks}\label{subsec:implication-ranks}
One of the most important implications of exchangeability of real-valued random variables is that the ranks of $W_1, \ldots, W_n$ (among this sequence) are uniformly distributed on $\{1,2,\ldots,n\}$. If $\mathcal{X} := \{x_1, \ldots, x_n\}$ is a set with $n$ elements (meaning that all the elements in $\mathcal{X}$ are distinct), then the rank of $x_i$ among $\mathcal{X}$ can be defined as
\begin{equation}\label{eq:rank-no-ties}
\mathrm{rank}^*(x_i; \mathcal{X}) ~:=~ \left|\{j\in[n]:\,x_j \le x_i\}\right|.
\end{equation}
In other words, rank of $x_i$ is the number of elements in $\mathcal{X}$ (including itself) that are smaller than or equal to $x_i$. 
If the set $\mathcal{X}$ has fewer than $n$ elements (meaning that some elements of $\mathcal{X}$ are equal), then ranks as defined in~\eqref{eq:rank-no-ties} will lead to ties, that is, the elements of $\mathcal{X}$ that are equal get the same rank. 

{\clr Ranking with ties would, in general, lead to a distribution of $(\mathrm{rank}(X_i; \mathcal{X}), i\in[n])$ that depends on the distribution of $X_i$.} For example, if we have a sequence of $n = 100$ i.i.d. Bernoulli$(p)$ random variables, then the number of 1's and 2's in the sequence of ranks defined by~\eqref{eq:rank-no-ties} depends on $p$. This causes a hindrance to the distribution-free prediction and non-parametric ranks. There are different ways of breaking ties in ranks obtained from~\eqref{eq:rank-no-ties}. For simplicity, we consider the following definition (from~\cite{vorlickova1972asymptotic}) of ranks that works for all sequences alike. 
\begin{defi}[Rank]\label{def:rank-definition}
For a set of real numbers $\mathcal{X} := \{x_1, \ldots, x_n\}$, define the rank of $x_i$ among $\mathcal{X}$ as
\begin{equation}\label{eq:rank-general}
\begin{split}
\mathrm{rank}(x_i; \mathcal{X}) ~&:=~ |\{j\in[n]:\,x_j + \xi U_j \le x_i + \xi U_i\}|\\ ~&=~ \mathrm{rank}^*(x_i + \xi U_i;\,\mathcal{X} + \xi\mathcal{U}),
\end{split}
\end{equation}
where $\xi > 0$ is arbitrary and $U_1, \ldots, U_n$ are iid $\mathrm{Unif}[-1, 1]$ random variables. Here $\mathcal{X} + \xi\mathcal{U} = \{x_i + \xi U_i:\, 1\le i\le n\}.$ 
\end{defi}
Because $U_1, \ldots, U_n$ are almost surely distinct, we get that $x_i + \xi U_i, i\in [n]$ are also distinct with probability one, irrespective of whether $x_1, \ldots, x_n$ have ties. Essentially, jittering the original sequence with some noise breaks the ties and brings the situation back to the case of no ties. This is important to obtain the distribution-free nature of the results to be described and with ties this would not be possible. For further convenience, we formally define the jittered sequence.
\begin{defi}[Jittered Sequence]\label{def:jittered-sequence}
For any $n\ge1$ and sequence $x_1, \ldots, x_n$ {\clr of real numbers}, the jittered sequence with parameter $\xi > 0$ is defined as $x_1^*, \ldots, x_n^*$ with
\[
x_i^* ~:=~ x_i + \xi U_i.
\] 
Here $U_1, \ldots, U_n$ are iid $\mathrm{Unif}[-1,1]$ random variables. We suppress $\xi$ in the notation of $x^*_i$ for convenience. 
\end{defi}
In general, the definition of rank above depends on $\xi > 0$. If $x_i, i\in[n]$ do not have ties, then $\mathrm{rank}(x_i; \mathcal{X})$ in~\eqref{eq:rank-general} matches the one in~\eqref{eq:rank-no-ties} as $\xi$ tends to zero.  For a general sequence with ties, the rank in~\eqref{eq:rank-general} breaks the ties for ranking randomly as $\xi$ tends to zero.  For the purposes of exchangeability, the size of $\xi$ is immaterial but in practice, fixing $\xi$ to be a small constant such as $10^{-8}$ relative to the spacings in the data works as expected for all sets.

Definition~\ref{def:rank-definition} of ranks coupled with the definition of exchangeability implies the following result proved in Appendix~\ref{appsec:EXCHANGEABILITY-IMPLICATIONS}.
\begin{thm}\label{thm:rank-distribution}
If $W_1, \ldots, W_n$ are exchangeable random variables, then for any $\xi > 0$,
\[
\big(\mathrm{rank}(W_i; \{W_1, \ldots, W_n\}):\,i\in[n]\big) ~\sim~ \mathrm{Unif}\left(\{\pi:[n]\to[n]\}\right).
\]
Here $\mathrm{Unif}\left(\{\pi:[n]\to[n]\}\right)$ represents the uniform distribution over all permutations of $[n]$, that is, each permutation has an equal probability of $1/n!$.
\end{thm} 
Theorem~\ref{thm:rank-distribution} shows that the ranks of $W_i, i\in[n]$ are exchangeable, and further that their distribution does not depend on the distribution of $W_i$. It should be mentioned here that the distribution of the ranks is computed including the randomness of $U_1, \ldots, U_n$; they are not conditioned on. This theorem also represents one of the most useful implications of exchangeability and is crucial in proving the validity of rank tests as well as conformal prediction. 

For the validity guarantees of rank tests, Theorem~\ref{thm:rank-distribution} in its form is enough. For the validity guarantees of conformal prediction, we need the following corollary (proved in Appendix~\ref{appsec:EXCHANGEABILITY-IMPLICATIONS}) of Theorem~\ref{thm:rank-distribution}.
\begin{cor}\label{cor:for-conformal-prediction}
Under the assumptions of Theorem~\ref{thm:rank-distribution}, for any $\xi > 0$, we have
\[
\mathbb{P}\bigg(\mathrm{rank}(W_n;\,\{W_1, \ldots, W_n\}) \le t\bigg) ~=~ \frac{\lfloor t\rfloor}{n}, 
\]
where, for $t\in\mathbb{R}$, $\lfloor t\rfloor$ represents the largest integer smaller than {\clr or equal to} $t$. {\clr Moreover, the random variable $P := \mathrm{rank}(W_n; \{W_1,\ldots,W_n\})/n$ is a valid $p$-value, i.e., 
\[
\mathbb{P}(P \le \alpha) ~\le~ \alpha\quad\mbox{for all}\quad \alpha\in[0, 1].
\]}
\end{cor}
\subsection{Transformations Preserving Exchangeability}\label{subsec:transformations}
Theorem~\ref{thm:rank-distribution} holds for real-valued random variables\footnote{For random variables in a metric space, the definition of ranks can be extended and a result similar to Theorem~\ref{thm:rank-distribution} can be proved; see~\cite{deb2019multivariate} for details.} and to explore the full strength of exchangeability in arbitrary spaces, we transform random variables from arbitrary spaces to the real line. {\clr If the transformation to the real line does not depend on the data (or is constructed from an independent data), then it is relatively easy to verify that the transformed variables also form an exchangeable sequence. In many cases, one might not have access to independent data or might want to use the full data for a more ``powerful'' transformation.} For such purposes, we need a result to verify exchangeability of random variables after transformation stated below. As a motivation, consider the following examples:
\begin{itemize}
    \item Suppose $W_1, \ldots, W_n$ are exchangeable. Consider the transformed variables $W_1 - \overline{W}_n, \ldots, W_n - \overline{W}_n$, where $\overline{W}_n$ is the average of the $n$ variables. In this case the transformation takes $n$ variables to $n$ variables. Intuitively, these are exchangeable but how does one prove it rigorously. One could use~\eqref{eq:cumulative-distribution-function}.
    \item {\clr In the same setting as above, consider the transformed variables to be} $W_1 - \overline{W}_{-3}, W_2 - \overline{W}_{-3}, W_3 - \overline{W}_{-3}$, where $\overline{W}_{-3}$ is the average of $W_4, \ldots, W_n$ (the sequence without the first three elements).    
    \item {\clr In the same setting as above, consider the transformed variables to be $W_1 - \overline{W}_{n-1}, \ldots, W_{n-1}-\overline{W}_{n-1}, W_n - \overline{W}_{n-1}$. In this case, the transformation depends only on first $n-1$ variables and takes a sequence of $n$ variables to $n$ variables.}
\end{itemize} 
The following {\clr is an} important result about transformations preserving exchangeability {\clr taken from~\cite{dean1990linear} and~\citet[Section 2.5]{commenges2003transformations}}. The setting is as follows: $W_1, \ldots, W_n$ are random variables taking values in a space $\mathcal{W}$ and $G$ is a transformation taking a vector of $n$ elements in $\mathcal{W}$ to a vector of $m$ elements in another space $\mathcal{W}'$. Here $\mathcal{W}$ and $\mathcal{W}'$ are arbitrary sets. (Usually $\mathcal{W}$ would be an arbitrary space and $\mathcal{W}'$ is the real line.) {\clr For any for $y = (y_1, \ldots, y_m)$ and a permutation $\pi_1:[m]\to[m]$, set $\pi_1y = (y_{\pi_1(1)}, \ldots, y_{\pi_1(m)})$.
\begin{thm}[\cite{dean1990linear}]\label{thm:transformations-exchangeability}
Suppose $W = (W_1, \ldots, W_n)\in\mathcal{W}^n$ is a vector of exchangeable random variables. Fix a transformation $G:\mathcal{W}^n\to(\mathcal{W}')^m$. If for each permutation $\pi_1:[m]\to[m]$ there exists a permutation $\pi_2:[n]\to[n]$ such that
\begin{equation}\label{eq:permutation-condition}
\pi_1G(w) ~=~ G(\pi_2w),\quad\mbox{for all}\quad w\in\mathcal{W}^n,
\end{equation}
then $G(\cdot)$ preserves exchangeability of $W$. Conversely, if $G(\cdot)$ preserves exchangeability of $W$ whatever the distribution of $W$, then for each permutation $\pi_1:[m]\to[m]$ and $w\in\mathcal{W}^n$, there exists a permutation $\pi_2:[n]\to[n]$ (possibly depending on $w$) such that $\pi_1G(w) = G(\pi_2w)$. Furthermore, if $G(\cdot)$ is a linear transformation, then $G$ is exchangeability preserving if and only if~\eqref{eq:permutation-condition} holds true.
\end{thm}
Theorem~\ref{thm:transformations-exchangeability} follows from the proof of Theorem 4 of~\cite{dean1990linear}. \citet[Section 2.5]{commenges2003transformations} states (without proof) that~\eqref{eq:permutation-condition} is a necessary and sufficient condition for $G(\cdot)$ to be exchangeability preserving. At present, we could only prove Theorem~\ref{thm:transformations-exchangeability}; see Appendix~\ref{appsec:commenges}. The only difference between the necessary and sufficient conditions in Theorem~\ref{thm:transformations-exchangeability} is that the permutation $\pi_2$ can depend on $w$ in the necessary condition.}
This theorem, in words, states that a transformation of exchangeable random variables is exchangeable if a permutation of the transformed random variables is equal to the transformation applied to a permutation of the original exchangeable random variables.

As an application, we will revisit the examples discussed above. 
\begin{itemize}
\item In the first example above, the transformation is
\[
G:\, (W_1, \ldots, W_n) ~\mapsto~ (W_1 - \overline{W}_n, \ldots, W_n - \overline{W}_n).
\]
A permutation of the right hand side is $$(W_{\pi(1)} - \overline{W}_n, \ldots, W_{\pi(n)} - \overline{W}_n),$$ and this is equal to $G(W_{\pi(1)}, \ldots, W_{\pi(n)})$, because the average of $n$ variables is a symmetric function and does not change with a permutation. Hence, $G(W_1, \ldots, W_n)$ is a vector of exchangeable random variables whenever $W_1, \ldots, W_n$ are exchangeable by Theorem~\ref{thm:transformations-exchangeability}. 
\item In the second example, the transformation is
\[
G:\,(W_1, \ldots, W_n) ~\mapsto~ (W_1 - \overline{W}_{-3}, W_2 - \overline{W}_{-3}, W_3 - \overline{W}_{-3}).
\]
In this case if we permute the right hand side to get $(W_{\pi(1)} - \overline{W}_{-3}, W_{\pi(2)} - \overline{W}_{-3}, W_{\pi(3)} - \overline{W}_{-3})$, then it corresponds to applying the same permutation on the first three elements of $W_1, \ldots, W_n$ and leaving the remaining elements as is. Formally, take $\pi_1:[n]\to[n]$ such that $\pi_1(1)=\pi(1), \pi_1(2) = \pi(2)$, $\pi_1(3) = \pi(3)$, and $\pi_1(i) = i$ for $i\ge4$. This implies that $W_1 - \overline{W}_{-3}, W_2 - \overline{W}_{-3}, W_3 - \overline{W}_{-3}$ are exchangeable if $W_1, \ldots, W_n$ are exchangeable, again by Theorem~\ref{thm:transformations-exchangeability}. {\clr This application is related to the split conformal method discussed in Section~\ref{sec:split-conformal}.}
\item {\clr For the third example, the transformation is
\[
G:\,(W_1, \ldots, W_n) ~\mapsto~ (W_1 - \overline{W}_{n-1}, \ldots, W_{n-1}-\overline{W}_{n-1}, W_n - \overline{W}_{n-1}).
\]
Note that this is a linear transformation and Theorem~\ref{thm:transformations-exchangeability} provides a necessary and sufficient condition.
If we apply a permutation $\pi$ on $G(W_1,\ldots,W_n)$, we get
\begin{equation}\label{eq:transformed-vector-example-3}
(W_{\pi(1)} - \overline{W}_{n-1},\, \ldots, W_{\pi(n-1)} - \overline{W}_{n-1},\, W_{\pi(n)} - \overline{W}_{n-1}).
\end{equation}
Because $\overline{W}_{n-1}$ is an asymmetric function of $(W_1, \ldots, W_n)$, the vector in~\eqref{eq:transformed-vector-example-3} is not equal to the transformation $G$ applied to $(W_{\pi(1)},\ldots, W_{\pi(n)})$. This implies that, in general, $G(W_1, \ldots, W_n)$ is not a vector of exchangeable random variables.}
\end{itemize}
Having described in details the implications of exchangeability, we now proceed to explore the applications for conformal prediction and rank tests. All the results that follow are corollaries of the results in the current section. {\clr It might be worth mentioning here that none of the results in the paper are new or difficult to prove. They are all standard.} This is the main intent of the article: to show that most of conformal prediction and non-parametric rank tests follow from some basic facts about exchangeability. 
\section{Conformal Prediction}\label{SEC:CONFORMAL-PREDICTION}
Conformal prediction is a {\clr generic tool for} finite sample, distribution-free valid predictive inference introduced by~\cite{vovk2005algorithmic} {\clr  and~\cite{shafer2008tutorial}}. This method of predictive inference was reintroduced to the statistics community by~\cite{MR3174619}. 
\subsection{Formulation of the Problem}
The general formulation of the prediction problem is as follows. Given realizations of $n$ exchangeable random variables $W_1, \ldots, W_n$, construct a prediction region for a future random variable, $W_{n+1},$ that is exchangeable with the first $n$ random variables, i.e., $W_1, \ldots, W_{n+1}$ is a sequence of exchangeable random variables. Mathematically, for $\alpha\in[0, 1]$, construct a prediction region $\widehat{\mathcal{R}}_{n,\alpha}$ depending on $W_1, \ldots, W_n$, that is,
\[
\widehat{\mathcal{R}}_{n,\alpha} = \widehat{\mathcal{R}}_{n,\alpha}(W_1, \ldots, W_n),
\] 
such that the $(n+1)$-st random variable $W_{n+1}$ belongs in this region with a probability of at least $1 - \alpha$:
\begin{equation}\label{eq:prediction-guarantee}
\mathbb{P}\left(W_{n+1} \in \widehat{\mathcal{R}}_{n,\alpha}\right) ~\ge~ 1 - \alpha,
\end{equation}
whenever $W_1, \ldots, W_{n+1}$ form a sequence of exchangeable random variables. {\clr In~\eqref{eq:prediction-guarantee}, the probability $\mathbb{P}(\cdot)$ is the probability of the event with respect to the joint distribution of $(W_1, \ldots, W_{n+1})$.} In the general formulation here, there is no restriction on the random variables $W_1, \ldots, W_{n+1}$ to be real-valued; in fact, they may be elements of an arbitrary sample space $\mathcal{W}$. {\clr Prediction problems in general spaces occur in applications such as functional data analysis~\citep{lei2015conformal} and image prediction~\citep{bates2021distribution}.}

The idea of conformal prediction for real-valued random variables would be that the rank of the future observation $W_{n+1}$ among the collection $\{W_1, \ldots, W_{n+1}\}$ is equally likely to be any of $1,2,\ldots,n+1$. We will deal with prediction in arbitrary spaces by using transformations that map these spaces to the real line, so that the rank transformation can be applied and the uniform distribution of the ranks can be leveraged (Section~\ref{subsec:implication-ranks}). To this end, Theorem~\ref{thm:transformations-exchangeability} would play an important role. 
\subsection{Full Conformal Prediction}\label{sec:full-conformal-method}
\subsubsection{Real-valued Random Variables.}
If $W_1, \ldots, W_{n+1}$ are real valued and exchangeable, then with ranks defined as in Definition~\ref{def:rank-definition}, Corollary~\ref{cor:for-conformal-prediction} implies that
\begin{align*}
&\mathbb{P}\bigg(\mathrm{rank}(W_{n+1};\,\{W_1, \ldots, W_{n+1}\}) \le \lceil(n+1)(1-\alpha)\rceil\bigg) ~=~ \frac{\lceil (n+1)(1 - \alpha)\rceil}{n+1}.
\end{align*}
It is easy to verify that the right hand side is at least $1 - \alpha$ and at most $1 - \alpha + 1/(n+1)$.  
Hence, a one-sided prediction region can be constructed as follows:
\begin{equation}\label{eq:conformal-prediction-main}
\widehat{\mathcal{R}}_{n,\alpha} ~:=~ \bigg\{w\in\mathbb{R}:\,\mathrm{rank}(w; \{W_1, \ldots, W_n, w\}) ~\le~ \lceil (n+1)(1 - \alpha) \rceil\bigg\},
\end{equation}
This is documented in the following result.
\begin{prop}\label{prop:conformal-prediction-main}
If $W_1, \ldots, W_{n+1}\in\mathbb{R}$ form a sequence of exchangeable random variables, then 
\[
1 - \alpha ~\le~ \mathbb{P}\left(W_{n+1}~\in~\widehat{\mathcal{R}}_{n,\alpha}\right) ~\le~ 1 - \alpha + \frac{1}{n+1},\quad\mbox{for all}\quad n\ge1, \alpha\in[0, 1],
\]
where the probability extends over all variables $W_1, \ldots, W_{n+1}$.
\end{prop}
This result is essentially proved above and follows from the basic corollary~\ref{cor:for-conformal-prediction} of the definition of exchangeability. Although the result is a restatement of Corollary~\ref{cor:for-conformal-prediction}, formulating the result in terms of prediction regions provides a form of finite sample distribution-free valid inference. Furthermore, the interval is not overly conservative in that the coverage is at most $1/(n+1)$ away from the required coverage of $(1 - \alpha)$. 
The set $\widehat{\mathcal{R}}_{n,\alpha}$ is defined implicitly and we now describe the computation of this prediction set.
\begin{tcolorbox}
\textbf{Pseudocode 1:} The set $\widehat{\mathcal{R}}_{n,\alpha}$ can be computed as follows.
\begin{enumerate}
    \item Take $\xi = 10^{-8}$ and generate $U_1, \ldots, U_n, U_{n+1}$ from the $\mathrm{Unif}[-1, 1]$ distribution. Define the jittered sequence
    $W_i^* := W_i + \xi U_i,\; i\in[n].$
    \item Sort the jittered random variables and let the sorted vector be
    \[
    W_{(1)}^* ~\le~ \ldots ~\le~ W_{(n)}^*.
    \]
    \item Compute $I := \lceil (n+1)(1-\alpha)\rceil$ and report the interval
    \[
    \left(-\infty,\, W_{(I)}^* - \xi U_{n+1}\right].
    \]
\end{enumerate}
\end{tcolorbox}
{\clr This is the simplest example of the full conformal method where exchangeability is invoked on the original set of real-valued random variables $W_1, \ldots, W_{n}, W_{n+1}$. Sometimes it might be useful to apply Proposition~\ref{prop:conformal-prediction-main} to a transformed data. For example, note that the prediction region $\widehat{\mathcal{R}}_{n,\alpha}$ is a one-sided interval and a two-sided interval might be preferable in practice.
\subsubsection{Arbitrary Spaces.}
For notational convenience, let $Z_1, Z_2, \ldots, Z_n\in\mathcal{Z}$ be exchangeable random variables and we want to construct a prediction region for $Z_{n+1}$, when $Z_1, \ldots, Z_n,$ $Z_{n+1}$ is a sequence of exchangeable random variables in $\mathcal{Z}$. 
Here $\mathcal{Z}$ can be $\mathbb{R}, \mathbb{R}^d,$ or any other arbitrary space.

For any (non-random) transformation $g:\mathcal{Z}\to\mathbb{R}$, 
\[
g(Z_1), g(Z_2), \ldots, g(Z_n), g(Z_{n+1}),
\]
are real-valued exchangeable random variables. This fact can be verified based on Theorem~\ref{thm:transformations-exchangeability}. Hence, Proposition~\ref{prop:conformal-prediction-main} applies and we obtain
\begin{equation}\label{eq:non-random-transformed-full-conformal}
\mathbb{P}\left(g(Z_{n+1}) \le (g(Z))^*_{(\lceil(n+1)(1-\alpha)\rceil)} - \xi U_{n+1}\right) \ge 1 - \alpha,
\end{equation}
where $(g(Z))^*_{(\lceil(n+1)(1-\alpha)\rceil)}$ is the $\lceil(n+1)(1-\alpha)\rceil$-th largest value of $g(Z_i) + \xi U_i, 1\le i\le n$. Inequality~\eqref{eq:non-random-transformed-full-conformal} yields a valid finite sample prediction region, irrespective of what $g:\mathcal{Z}\to\mathbb{R}$ is. 

For concrete examples, one can consider the following transformations. If $\mathcal{Z} = \mathbb{R}$, taking $g(z) = |z|$ yields the prediction region
\begin{equation}\label{eq:symmetric-real-set}
\{z\in\mathbb{R}:\,|z| \le |Z|^*_{(\lceil(n+1)(1-\alpha)\rceil)} - \xi U_{n+1}\}.
\end{equation}
This is a two-sided interval centered at $0$. 
If $\mathcal{Z} = \mathbb{R}^d$ or a general normed linear space with norm $\|\cdot\|$, then taking $g(z) = \|z\|$ yields the prediction region
\begin{equation}\label{eq:symmetric-normed-set}
\{z\in\mathcal{Z}:\,\|z\| \le \|Z\|^*_{(\lceil(n+1)(1-\alpha)\rceil)} - \xi U_{n+1}\}.
\end{equation}
This is a symmetric bounded set centered at $0\in\mathcal{Z}$. Prediction sets~\eqref{eq:symmetric-real-set} and~\eqref{eq:symmetric-normed-set} both suffer from the same disadvantage: they are symmetric around zero. If the true distribution of $Z_i$'s has a mode at some $z_0 \neq 0$, then these prediction regions are unnecessarily large. For instance, for the normal distribution with mean $1$, the shortest prediction interval is centered at $1$. However, a valid symmetric prediction interval centered at $0$ is about $1.34$ times larger than the shortest interval centered at $1$.

This disadvantage can be rectified by considering a data-dependent transformation. Now we need to consider transformations that retain exchangeability and hence, Theorem~\ref{thm:transformations-exchangeability} plays a crucial role. A data-dependent transformation $g:\mathcal{Z}\to\mathbb{R}$ denoted by $g(z;\, Z_1, \ldots, Z_n, Z_{n+1})$ is said to be permutation invariant if for any permutation $\pi:[n+1]\to[n+1]$,
\begin{equation}\label{eq:permutation-invariance-full-conformal}
g(z;\, Z_1, \ldots, Z_n, Z_{n+1}) = g(z;\, Z_{\pi(1)}, Z_{\pi(2)}, \ldots, Z_{\pi(n)}, Z_{\pi(n+1)}).
\end{equation}
For intuition, one may consider the following examples of permutation invariant transformations.
\begin{itemize}[leftmargin=*]
\item \textbf{Location Centering:} If $\mathcal{Z} = \mathbb{R}$, then an example transformation is
\[
g(z; Z_1, \ldots, Z_{n}, Z_{n+1}) ~:=~ \left|z - \frac{1}{n+1}\sum_{j=1}^{n+1} Z_j\right|. 
\]
Because the mean is permutation invariant, this is a permutation invariant transformation. Clearly, one can replace the mean by the median of $Z_1, \ldots, Z_{n+1}$, which is also a symmetric function. If $\mathcal{Z}$ is a normed linear space with the norm $\|\cdot\|$, then define
\[
g(z; Z_1, \ldots, Z_{n}, Z_{n+1}) ~:=~ \left\|z - \frac{1}{n+1}\sum_{j=1}^{n+1} Z_j\right\|. 
\]
Once again, this is also permutation invariant.
\item \textbf{Density Transformation:} If $\mathcal{Z} = \mathbb{R}$, then define
\begin{equation}\label{eq:density-transformation}
g(z;\, Z_1, \ldots, Z_n, Z_{n+1}) ~:=~ \frac{1}{\widehat{p}_{n+1}(z)},
\end{equation}
where $\widehat{p}_{n+1}(\cdot)$ is a density estimator that depends permutation invariantly on $Z_1, \ldots, Z_{n+1}$. For example, one can take $\widehat{p}_{n+1}(\cdot)$ to be the kernel density estimator
\[
\widehat{p}_{n+1}(s) ~:=~ \frac{1}{(n+1)h}\sum_{i=1}^{n+1} k\left(\frac{s - Z_j}{h}\right),
\]
for a kernel function $k(\cdot)$ and bandwidth $h > 0$. A similar density estimator can be constructed in normed spaces. The transformation~\eqref{eq:density-transformation} was considered in~\cite{MR3174619} to construct asymptotically optimal prediction sets in $\mathbb{R}^d$. Here optimality is in terms of smallest volume or Lebesgue measure.
\item \textbf{Regression Residual:} If $\mathcal{Z} = \mathbb{R}\times\mathcal{X}$ and $Z_i = (Y_i, X_i), 1\le i\le n+1$ for a regression data, then an example transformation targeting the response is
\begin{equation}\label{eq:regression-residual}
g(z; Z_1, \ldots, Z_n, Z_{n+1}) ~:=~ |y - \widehat{\mu}(x; Z_1, \ldots, Z_{n+1})|,
\end{equation}
for $z = (y, x)$. Here $\widehat{\mu}(\cdot; Z_1, \ldots, Z_{n+1})$ represents a non-parametric regression mean estimator that depends permutation invariantly on $Z_1, \ldots, Z_{n+1}$. For example, it can be the kernel regression estimator
\[
\widehat{\mu}(x; Z_1, \ldots, Z_{n+1}) ~:=~ \frac{\sum_{j=1}^{n+1} Y_jk((x - X_j)/h)}{\sum_{j=1}^{n+1} k((x - X_j)/h)},
\]
for a kernel function $k(\cdot)$ and bandwidth $h > 0$. The transformation~\eqref{eq:regression-residual} leads to a non-trivial prediction region for the response but a trivial one for the predictors. This feature will be discussed later in Section~\ref{appsec:conformal-regression}.   
\end{itemize}
Under the permutation invariance condition~\eqref{eq:permutation-invariance-full-conformal}, Theorem~\ref{thm:transformations-exchangeability} implies that when $z = Z_{n+1}$, $$W_i(z) := g(Z_i;\,Z_1, \ldots, Z_n, z),\,1\le i\le n,\; W_{n+1}(z) := g(z;\,Z_1, \ldots, Z_n, z),$$ is a sequence of exchangeable real-valued random variables; see Proposition~\ref{prop:permutation-invariant-exchangeable} (of Appendix~\ref{appsec:auxiliary-results}) for a formal result. Hence applying Proposition~\ref{prop:conformal-prediction-main} to $W_i(z)$'s, we obtain the prediction region
\begin{equation}\label{eq:full-conformal-general-set}
\{z\in\mathcal{Z}:\, \mathrm{rank}(W_{n+1}(z);\, \{W_1(z), \ldots, W_n(z), W_{n+1}(z)\}) \le \lceil(n+1)(1-\alpha)\rceil\}.
\end{equation}
A distinguishing feature of this prediction region compared to the one from Pseudocode 1, and in~\eqref{eq:symmetric-real-set},~\eqref{eq:symmetric-normed-set}, is that there is no closed form expression. Note that all the $W_i$'s depend on the unknown $z = Z_{n+1}$. Computing the region~\eqref{eq:full-conformal-general-set}, in general, requires computing $g(\cdot;\, Z_1, \ldots, Z_n, z)$ for all $z\in\mathcal{Z}$ and then verifying the rank condition in~\eqref{eq:full-conformal-general-set}. Hence, the prediction region~\eqref{eq:full-conformal-general-set} is, in general, computationally inefficient. It should be mentioned that there do exist cases where the full conformal prediction set~\eqref{eq:full-conformal-general-set} can be computed efficiently; see~\cite{burnaev2014efficiency}, \cite{chen2018discretized}, \cite{lei2019fast}, and~\cite{ndiaye2019computing} for some examples.  

Because of the heavy computational cost of the full conformal method, we now focus on the split conformal method that provides the same finite sample validity guarantees and is computationally efficient.
}

\subsection{Split Conformal Prediction}\label{sec:split-conformal}
\subsubsection{Real-valued Random Variables.}
Following~\cite{papadopoulos2002inductive} and~\cite{lei2018distribution}, we now discuss a split conformal prediction method which can be used to construct efficient prediction regions without the computational burden of the full conformal method. The procedure, in words, is as follows. We split the exchangeable sequence into two parts, and from the first part, we compute the average. Then the variables in the second part (along with the future variable) centered at the average of the first part are exchangeable, which leads us to a two-sided prediction interval. Formally,
given random variables $Z_1, \ldots, Z_n\in\mathbb{R}$ and $n_1 \in [n]$, construct the split
\begin{equation}\label{eq:splitting}
\begin{split}
\mathcal{T} ~&:=~ \{Z_1, \ldots, Z_{n_1}\},\quad\mbox{and}\quad
\mathcal{C} ~:=~ \{Z_{n_1 + 1}, \ldots, Z_{n}\}.
\end{split}  
\end{equation}
(The notations $\mathcal{T}$ and $\mathcal{C}$ stand for training and calibration sets, respectively.) From the training set, compute $\widebar{Z}_{\mathcal{T}}$, the average of the observations in~$\mathcal{T}$. From the calibration set, compute the random variables
$|Z_{n_1 + 1} - \widebar{Z}_{\mathcal{T}}|, \ldots, |Z_{n} - \widebar{Z}_{\mathcal{T}}|.$
Proposition~\ref{prop:split-exchangeable} (of Appendix~\ref{appsec:auxiliary-results}) implies that these random variables are exchangeable with $|Z_{n+1} - \widebar{Z}_{\mathcal{T}}|$. Now applying Proposition~\ref{prop:conformal-prediction-main} with 
\begin{equation}\label{eq:mean-centering-split-conformal}
W_1 := |Z_{n_1 + 1} - \widebar{Z}_{\mathcal{T}}|,\; \ldots,\; W_{n-n_1} := |Z_{n} - \widebar{Z}_{\mathcal{T}}|
\end{equation}
yields the prediction region: 
\begin{equation}\label{eq:real-split-conformal-mean}
\widehat{\mathcal{R}}_{n,\alpha}^{\mathrm{split}} ~:=~ \left\{z\in\mathbb{R}:\, |z - \widebar{Z}_{\mathcal{T}}| + \xi U_{n+1} \le W_{i_{\lceil (n-n_1+1)(1-\alpha)\rceil}}^*\right\}.
\end{equation}
Here $U_{n+1}\sim U[0, 1]$, $W^*$ represents the jittered sequence in Definition~\ref{def:jittered-sequence}, and $W_{(\lceil (n-n_1+1)(1-\alpha)\rceil)}^*$ represents the $\lceil (n-n_1+1)(1-\alpha)\rceil$-th largest value among $W_1^*, \ldots, W_{n-n_1}^*$.
\begin{prop}\label{lem:split-conformal-basic}
If $Z_1, \ldots, Z_{n+1}\in\mathbb{R}$ form a sequence of exchangeable random variables, then for all $n \ge n_1\ge1$ and $\alpha\in[0, 1]$,
\[
1 - \alpha ~~\le~~ \mathbb{P}\left(Z_{n+1}~\in~\widehat{\mathcal{R}}_{n,\alpha}^{\mathrm{split}}\right) ~~\le~~ 1 - \alpha + \frac{1}{n-n_1+1},
\] 
where the probability extends over all variables, including $Z_1, \ldots, Z_n$ used to construct $\widehat{\mathcal{R}}_{n,\alpha}^{\mathrm{split}}.$
\end{prop}
Although the mean $\widebar{Z}_{\mathcal{T}}$ is a natural choice in~\eqref{eq:real-split-conformal-mean} given the fact that the average is in some sense the best point predictor of a random variable,\footnote{The value $a\in\mathbb{R}$ that minimizes $\mathbb{E}[(Z - a)^2]$ (the prediction risk of $a$) is given by $\mathbb{E}[Z]$.} Proposition~\ref{lem:split-conformal-basic} continues to hold true if $\widebar{Z}_{\mathcal{T}}$ is  generalized by replacing it with $h(Z_1, \ldots, Z_{n_1})$ for any function $h$ of $Z_1, \ldots, Z_{n_1}$. This again is because of Proposition~\ref{prop:split-exchangeable}.
\subsubsection{Arbitrary Spaces.}

Suppose $Z_1, Z_2, \ldots, Z_{n+1}\in\mathcal{Z}$ are exchangeable random variables. Here $\mathcal{Z}$ can be a space of functions or a space of images or a space of documents. For any (non-random) transformation $g:\mathcal{Z}\to\mathbb{R}$,
\begin{equation}\label{eq:trasnformed-exchangeable}
g(Z_1),\, g(Z_2),\, \ldots,\, g(Z_n),\, g(Z_{n+1}),
\end{equation}
are real valued exchangeable random variables. Hence Proposition~\ref{lem:split-conformal-basic} applies to $W_i = g(Z_{n_1+i}), 1\le i\le n-n_1$ and leads to prediction regions
\[
\widehat{\mathcal{R}}_{n,\alpha}^{\mathrm{split}} ~:=~ \bigg\{z\in\mathcal{Z}:\,|g(z) - \widebar{g(Z)}_{\mathcal{T}}| \le W_{(\lceil (n - n_1 + 1)(1-\alpha)\rceil)}^* - \xi U_{n+1}\bigg\}.
\]
{\clr Here again $W^*_{(\lceil(n-n_1+1)(1-\alpha))\rceil}$ is the $\lceil(n-n_1+1)(1-\alpha)\rceil$-th largest value of $W_1, \ldots, W_{n-n_1}$.} 
It is noteworthy that for the prediction coverage validity of $\widehat{\mathcal{R}}_{n,\alpha}^{\mathrm{split}}$, the only requirement is that the random variables $g(Z_1), \ldots, g(Z_{n+1})$ are exchangeable and this can hold for some data-driven transformations. 
More precisely, the analyst is allowed to use the training part of the data to construct a transformation $\widehat{f}_{\mathcal{T}}(\cdot)$. Then exchangeability of $Z_1, \ldots, Z_{n+1}$ implies that $\widehat{f}_{\mathcal{T}}(Z_{n_1 + 1}), \ldots, \widehat{f}_{\mathcal{T}}(Z_{n+1})$ are exchangeable and can be used for constructing a prediction interval for $Z_{n+1}$.
This leads to the following result proved in Appendix~\ref{appsec:conformal-prediction}. (This result is given the status of a theorem because it shows the generality of conformal prediction.) Recall here that $\mathcal{Z}$ denotes the space in which the random variables $Z_1, \ldots, Z_n$ lie.
\begin{thm}\label{thm:main-split-conformal}
For any $\widehat{f}_{\mathcal{T}}:\mathcal{Z}\to\mathbb{R}$, an arbitrary function depending only on $\mathcal{T}$, define
\begin{equation}\label{eq:split-conformal-most-general}
\widehat{\mathcal{R}}_{n,\alpha}^{\mathrm{split}} ~:=~ \left\{z\in\mathcal{Z}:\,\widehat{f}_{\mathcal{T}}(z) ~\le~ (\widehat{f}_{\mathcal{T}}(Z))^*_{(\lceil (n-n_1+1)(1-\alpha)\rceil)} - \xi U_{n+1}\right\},
\end{equation}
where $(\widehat{f}_{\mathcal{T}}(Z))^*_{(\lceil (n-n_1+1)(1-\alpha)\rceil)}$ is the $\lceil(n-n_1+1)(1-\alpha)\rceil$-th largest value among the jittered sequence $(\widehat{f}_{\mathcal{T}}(Z_i))^*, i\in\{n_1+1,\ldots,n\}$. Then
\[
1 - \alpha ~~\le~~ \mathbb{P}\left(Z_{n+1}~\in~\widehat{\mathcal{R}}_{n,\alpha}^{\mathrm{split}}\right) ~~\le~~ 1 - \alpha + \frac{1}{n-n_1 + 1}.
\]
\end{thm}
In words, Theorem~\ref{thm:main-split-conformal} implies that if we find a real-valued transformation based on the training set and construct a prediction region based on the transformed calibration set, then it is also a valid $(1-\alpha)$ prediction, under exchangeability of the set of random variables $Z_1, \ldots, Z_{n+1}$.

{\clr  Theorem 3.1 is not new and is well-known in the conformal prediction literature. See, for example, Theorem 2 of~\cite{lei2018distribution}. In the language of conformal scores, the transformation $\widehat{f}_{\mathcal{T}}(Z_i)$ would be a conformal score corresponding to $Z_i$. As can be seen from Theorem 3.1, no specific properties of this transformation are required for validity. Because we are constructing a prediction set that contains smaller values of $\widehat{f}_{\mathcal{T}}(\cdot)$, the set $\widehat{\mathcal{R}}_{n,\alpha}^{\mathrm{split}}$ would be sensible if a smaller value of $\widehat{f}_{\mathcal{T}}(z)$ corresponds to $z$ ``conforming'' with the training data. For instance, the split conformal prediction set~\eqref{eq:real-split-conformal-mean} is constructed based on the conformal score $\widehat{f}_{\mathcal{T}}(z) = |z - \widebar{Z}_{\mathcal{T}}|$. A smaller value here means $z$ is close to the ``center'' of training data --- conforming. A larger value means $z$ is away from the training data --- not conforming.}
\begin{tcolorbox}
\noindent\textbf{Pseudocode 2:} Computationally, Theorem~\ref{thm:main-split-conformal} works in practice as follows:
\begin{enumerate}
    \item Split the data $Z_1, \ldots, Z_n$ into two parts: $\mathcal{T}$ and $\mathcal{C}$ as in~\eqref{eq:splitting}.
    \item Based on the training data $\mathcal{T}$, find a transformation $\widehat{f}_{\mathcal{T}}:\mathcal{Z}\to\mathbb{R}$. See Section~\ref{subsec:examples-conformal} for some examples of $\widehat{f}_{\mathcal{T}}$.
    \item Take $\xi = 10^{-8}$ and generate $U_1, \ldots, U_{n+1}$ from $\mathrm{Unif}[-1,1]$ distribution. Define the jittered sequence (based on the elements in the calibration set) as
    $W_i^* := \widehat{f}_{\mathcal{T}}(Z_{i}) + \xi U_i$, $n_1+1 \le i \le n.$
    \item Sort the observations $W_i$ as
    $W_{(1)}^* \le W_{(2)}^* \le \cdots \le W_{(n-n_1+1)}^*.$    
    \item Report the prediction region
    \[
    \left\{z\in\mathcal{Z}:\,\widehat{f}_{\mathcal{T}}(z) ~\le~ W_{(\lceil (n-n_1+1)(1-\alpha)\rceil)}^* - \xi U_{n+1}\right\}.
    \]
\end{enumerate}
\end{tcolorbox}
\subsubsection{Some Concrete Examples.}\label{subsec:examples-conformal} In the following, we present a few concrete examples/applications of Theorem~\ref{thm:main-split-conformal}. Note that the problem is still prediction: given $Z_1, \ldots, Z_n$, we want to predict $Z_{n+1}$. In the context of Theorem 3.1, we are essentially doing this based on $\widehat{f}_{\mathcal{T}}(Z_{n_1+1}), \ldots, \widehat{f}_{\mathcal{T}}(Z_{n})$. In the following examples, we describe the construction of some useful $\widehat{f}_{\mathcal{T}}(\cdot)$. The construction of the prediction region is based on Theorem~\ref{thm:main-split-conformal}. Recall here that $\mathcal{Z}$ denotes the space in which the random variables $Z_1, \ldots, Z_n$ lie.
\begin{enumerate}[leftmargin=*]
  \item \textbf{Norm-ball around the Mean:} Suppose $\mathcal{Z}=\mathbb{R}^d$ and let $\widebar{Z}_{\mathcal{T}}$ be the average of observations in $\mathcal{T}$, the training set. Take $\widehat{f}_{\mathcal{T}}(Z_i) = \|Z_i - \widebar{Z}_{\mathcal{T}}\|$ for $i\in\{n_1+1,\ldots,n\}$. Here $\|\cdot\|$ can be any semi-norm in $\mathbb{R}^d$; for instance the Euclidean norm or the Manhattan norm or the $\ell_p$-norm or even the absolute value of a single coordinate. 
  In the multivariate case, calculating the norm may not be meaningful if different coordinates of $Z$ have different units. If this is the case, then one can compute the norm of ``whitened'' vectors. More precisely, let $\widehat{\Sigma}_{\mathcal{T}}$ represent the sample covariance based on $\mathcal{T}$ and take $\widehat{f}_{\mathcal{T}}(Z_i) := \|\widehat{\Sigma}_{\mathcal{T}}^{-1/2}(Z_i - \widebar{Z}_{\mathcal{T}})\|$. If $\widehat{\Sigma}_{\mathcal{T}}$ is not invertible (which can happen if dimension is larger than $n_1$), then take $$\widehat{f}_{\mathcal{T}}(Z_i) := \|\mathrm{diag}(\widehat{\Sigma}_{\mathcal{T}})^{-1/2}(Z_i - \widebar{Z}_{\mathcal{T}})\|.$$
  Here $\mathrm{diag}(\widehat{\Sigma}_{\mathcal{T}})$ is the diagonal matrix corresponding to $\widehat{\Sigma}_{\mathcal{T}}.$ 
  \item \textbf{Principal Component Analysis (PCA):} In the previous example, the function $\widehat{f}_{\mathcal{T}}$ uses all the coordinates of $Z_i$ with no regard to the coordinates that matter more. For cases where the distribution of $Z_i$'s is supported on a low-dimensional manifold, this might be wasteful. One way to account for low-dimensionality is by using a dimension reduction technique on $\mathcal{T}$. For one concrete example, when $\mathcal{Z} = \mathbb{R}^d$, apply PCA on $\mathcal{T} = \{Z_1, \ldots, Z_{n_1}\}$ and fix $1\le \widehat{k}\le d$ to be the number of principal components (PCs) to be used. (This choice of $\widehat{k}$ can be based on any rule as long as it depends only on $\mathcal{T}$.) Let the first $\widehat{k}$ PCs be written into a matrix $\Pi\in\mathbb{R}^{\widehat{k}\times d}$ and consider 
  \[
  \widehat{f}_{\mathcal{T}}(Z_i) ~:=~ \|\mathrm{diag}(\widehat{\Sigma}_{\Pi, \mathcal{T}})^{-1/2}(\Pi Z_i - \widebar{(\Pi Z)}_{\mathcal{T}})\|.
  \]
  Here $\widehat{\Sigma}_{\Pi, \mathcal{T}}$ and $\widebar{(\Pi Z)}_{\mathcal{T}}$ represent the sample covariance matrix and sample average of $\Pi Z_1, \ldots, \Pi Z_{n_1}$. The semi-norm $\|\cdot\|$ above is arbitrary as in the previous example. Similar to the previous example, PCA is not special here and any of the many existing dimension reduction (linear or non-linear) techniques~\citep{cunningham2008dimension,xie2017survey,sorzano2014survey,nguyen2019ten,hinton2006reducing,wang2014generalized,tenenbaum2000global,silva2003global} can be used to get $\Pi Z_i$. In the context of functional data,~\cite{lei2015conformal} propose a few examples of $\Pi.$
  \item \textbf{Level Sets:} The prediction regions mentioned in the discussions before are all convex sets; in $\mathbb{R}$, these are intervals. These may, however, not be the optimal ones. For example, if the true distribution of $Z_i$ is a mixture of $N(3, 1)$ and $N(7, 1)$, then the optimal prediction region is a union of two intervals centered at the two modes; see Figure~\ref{fig:density_mixture_normals}. This is similar to the definition of high density regions, popular in Bayesian statistics; see~\cite{hyndman1996computing} for details.
  \begin{figure}[!h]
  \centering
  \includegraphics[height=2.7in,width=0.8\textwidth]{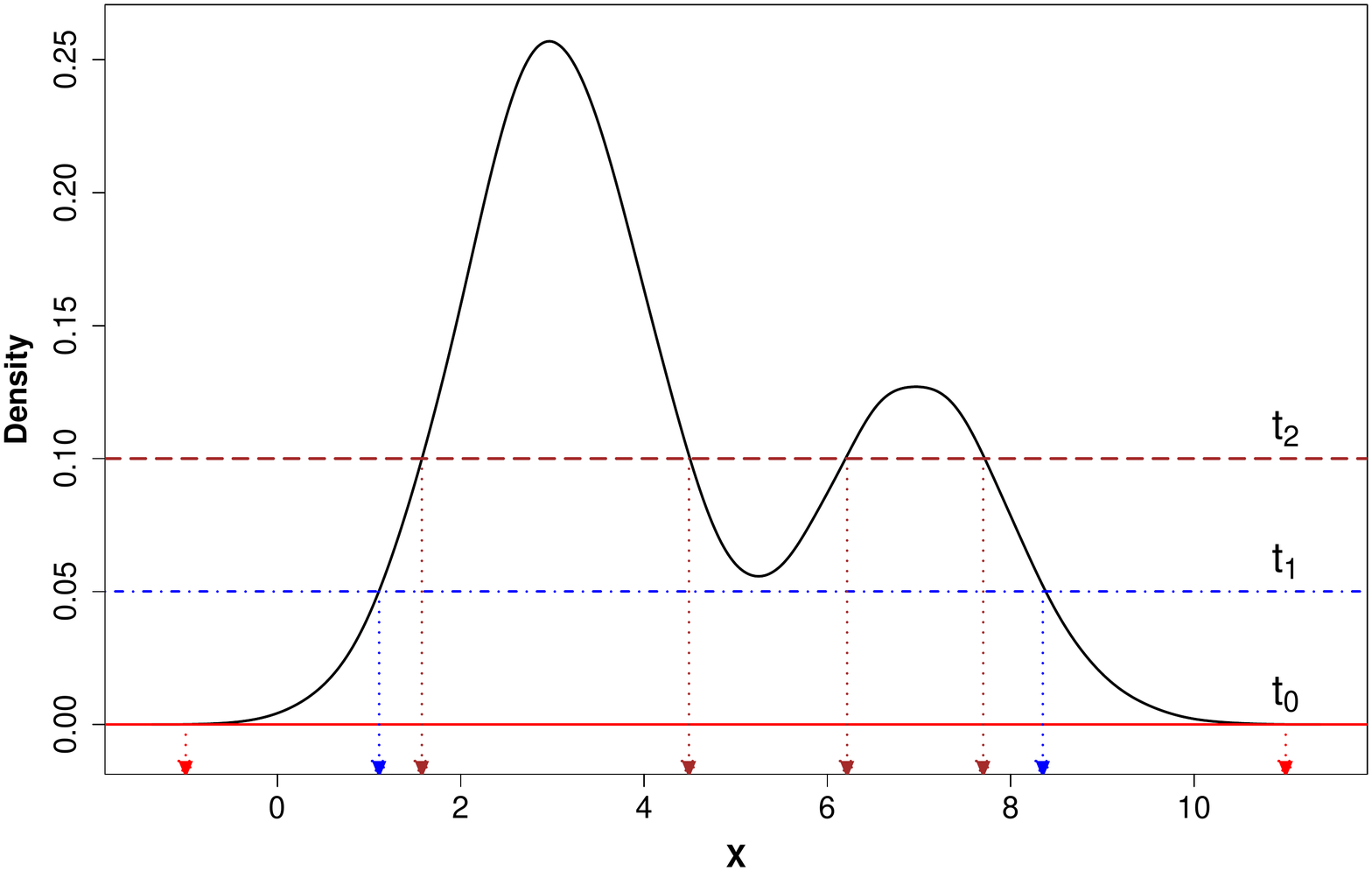}
  \caption{Illustration of level sets with the density of mixture of Gaussians centered at $3$ and $7$.}
  \label{fig:density_mixture_normals}
  \end{figure} 
  If the density of $Z_i$ is $p(\cdot)$, then the (oracle) optimal prediction region is given by
  \begin{equation}\label{eq:optimal-prediction-region}
  \mathcal{R}^{\mathrm{opt}} := \{z\in\mathcal{Z}:\,p(z) \ge t_{\alpha}\}, 
  \end{equation}
  where $t_{\alpha}$ is the largest $t$ solving $\int_{z:p(z) \ge t} p(z)dz \ge 1 - \alpha$. Sets of this type are called level sets. Of course, in practice we do not know the density $p$ and it might not even exist (with respect to the Lebesgue measure). One way to imitate this optimal prediction region is by taking
  \begin{equation}\label{eq:optimal-prediction-region-est}
  \widehat{f}_{\mathcal{T}}(Z_i) := 1/\widehat{p}_{\mathcal{T}}(Z_i),
  \end{equation}
  where $\widehat{p}_{\mathcal{T}}(\cdot)$ is an estimate of the density based on $Z_1, \ldots, Z_{n_1}$. Combining with the previous example, the density estimator here can be performed after an initial dimension reduction step. Furthermore, any of the existing density estimation methodologies (along with tuning parameter selection methods) can be used. The validity guarantees of conformal prediction does not depend on the estimation accuracy of the dimension reduction or density estimation methods. It can be proved that~\eqref{eq:optimal-prediction-region-est} along with Theorem~\ref{thm:main-split-conformal} leads to prediction regions that ``converge'' to the optimal one~\eqref{eq:optimal-prediction-region}; see~\cite{MR3174619} for details. 
\end{enumerate}
In some of the examples above, we have made an assumption that $\mathcal{Z}$ (the space in which $Z_1, \ldots, Z_n$ lie) is the Euclidean space $\mathbb{R}^d$. Some real data examples, such as image classification or topic modeling or text mining, do not satisfy this assumption readily. In all these examples, however, classical machine learning algorithms first convert the image or text data into a high-dimensional real-valued vector. Dimension is not an issue for conformal prediction because the validity is finite sample.
\subsection{Conformal Prediction for Regression}\label{appsec:conformal-regression}
{\clr In previous subsections, we have discussed the problem of prediction with no side information, i.e., we do not have any information at the future random variable. Most classical prediction algorithms in machine learning and statistics have covariate information for future random variable and the response is to be predicted. This forms one of the most interesting applications of conformal prediction}. Here the information in each observation is in two parts: covariates or predictors or features ($X$) and the response or class ($Y$). Let $Z_1 = (X_1, Y_1), \ldots, Z_n = (X_n, Y_n)$ be $n$ observations from a space $\mathcal{X}\times\mathcal{Y}$ and we want to predict $Y_{n+1}$ for $X_{n+1}\in\mathcal{X}$, whenever $(X_{n+1}, Y_{n+1})$ is exchangeable with $(X_i, Y_i), 1\le i\le n$. Formally, the goal is to construct $\widehat{\mathcal{R}}_{n,\alpha}(X_{n+1})$ such that
\begin{equation}\label{eq:marginal-conformal-regression}
\mathbb{P}\left(Y_{n+1} \in \widehat{\mathcal{R}}_{n,\alpha}(X_{n+1})\right) ~\ge~ 1 - \alpha.
\end{equation}
The probability on the left hand side is with respect to $(X_{n+1}, Y_{n+1})$ and also with respect to $(X_i, Y_i), i\in[n]$.
\vspace{0.1in}

\noindent\textbf{Marginal versus Conditional Coverage.} Mathematically, there is nothing wrong with the formulation~\eqref{eq:marginal-conformal-regression}, however, it is notationally misleading because of the alternative goal
\begin{equation}\label{eq:conditional-conformal-regression}
\mathbb{P}\left(Y_{n+1}\in\widehat{\mathcal{R}}_{n,\alpha}(X_{n+1})~\big|~X_{n+1} = x\right) \ge 1 - \alpha,
\end{equation}
for all $x\in\mathcal{X}$. Formulation~\eqref{eq:marginal-conformal-regression} provides a prediction region that covers $Y_{n+1}$ whenever $(X_{n+1}, Y_{n+1})$ comes from the same distribution. Formulation~\eqref{eq:conditional-conformal-regression}, however, requires a prediction region that covers $Y_{n+1}$ whatever the value of $X_{n+1}$ is. To understand the philosophical difference between these two goals, consider the following scenario. Suppose we have bivariate classification data with body mass index (BMI) as the covariate and indicator for the presence of cancer cells as the response. The goal~\eqref{eq:conditional-conformal-regression} provides a region such that whoever the next patient is, his/her response would lie in the constructed region with probability of at least $1 - \alpha$. {\clr If we get 100 new patients, then~\eqref{eq:conditional-conformal-regression} implies that among the patients with BMI$ = 15$ (say), the proportion of patients for which the true response lies in the constructed prediction set is about $1-\alpha$ and the same is true among the patients with BMI$ = 15.7, 16,$ and so on (any real number).} On the other hand, the goal~\eqref{eq:marginal-conformal-regression} cannot guarantee this but only implies that if we have 100 new patients and we give a region for each patient, then out of these 100 patients, for about $(1-\alpha)$ proportion of them the true response lies in their corresponding region. {\clr Importantly, among the patients with BMI$ =15$ (say), there is no specific control on the proportion of patients for which the true response lies in the constructed prediction set.}

It is easy to show that~\eqref{eq:conditional-conformal-regression} implies~\eqref{eq:marginal-conformal-regression}. It turns out that~\eqref{eq:conditional-conformal-regression} is too ambitious a goal in that it cannot be attained, non-trivially, in finite samples in a distribution-free setting; see~\citet[Section 2.6.1]{balasubramanian2014conformal} and~\cite{barber2019limits}. The latter reference discusses alternative conditional goals that are attainable sensibly. 
In this section, we will restrict attention to marginal validity as in~\eqref{eq:marginal-conformal-regression} and refer to~\cite{barber2019limits} for details on conditional validity~\eqref{eq:conditional-conformal-regression}.

\vspace{0.1in}
\noindent\textbf{Cross-sectional Conformal Method.} One simple way to attain the guarantee~\eqref{eq:marginal-conformal-regression} is as follows. Based on $(X_1, Y_1)$, $\ldots$, $(X_n, Y_n)\in\mathcal{X}\times\mathcal{Y}$, construct a prediction region $\widehat{\mathcal{R}}_{n,\alpha}\subseteq\mathcal{X}\times\mathcal{Y}$. Any of the constructions mentioned in Subsection~\ref{sec:split-conformal} can be used. These regions satisfy
\begin{equation}\label{eq:region-joint}
\mathbb{P}\left((X_{n+1}, Y_{n+1})\in \widehat{\mathcal{R}}_{n,\alpha}\right) ~\ge~ 1 - \alpha,
\end{equation}
and hence
\begin{equation}\label{eq:region-y-separate}
\mathbb{P}\left(Y_{n+1} \in \widehat{\mathcal{R}}_{n,\alpha}(X_{n+1})\right) ~\ge~ 1 - \alpha,
\end{equation}
where
\[
\widehat{\mathcal{R}}_{n,\alpha}(X_{n+1}) := \{y\in\mathcal{Y}:\,(X_{n+1}, y)\in\widehat{\mathcal{R}}_{n,\alpha}\}.
\]
This set is a cross-section of $\widehat{\mathcal{R}}_{n,\alpha}$ at $X_{n+1}$. See Figure~\ref{fig:cross-section} for an illustration of this method of attaining~\eqref{eq:marginal-conformal-regression}.
\begin{figure}[!h]
\centering
\hspace{-0.2in}\includegraphics[height=3in]{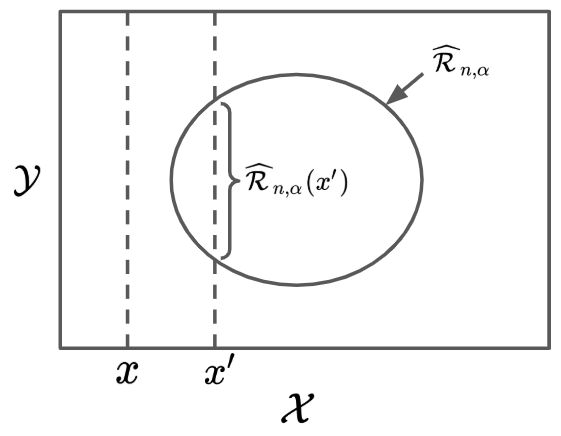}
\caption{Illustration of Cross-sectional Conformal Prediction for Regression and Classification.}
\label{fig:cross-section}
\end{figure}
In light of the discussion regarding conditional and marginal validity, we prefer writing~\eqref{eq:region-joint} instead of~\eqref{eq:region-y-separate} for clarity. The notation~\eqref{eq:region-joint} makes it clear that the construction of the region has nothing to do with the realized value of $X_{n+1}$. {\clr This also shows that there is no formal difference between conformal prediction for regression and conformal prediction for general spaces. The only difference is in the construction of the transformation; now the transformation focuses on how well the response value conforms with the training data for the realized value of $X_{n+1}$. }

Figure~\ref{fig:cross-section} shows a commonly mentioned disadvantage of this cross-sectional method. For some values of $X_{n+1}$, $\widehat{\mathcal{R}}_{n,\alpha}(X_{n+1})$ can be empty. For example, in Figure~\ref{fig:cross-section}, $\widehat{\mathcal{R}}_{n,\alpha}(x) = \emptyset$ and $\widehat{\mathcal{R}}_{n,\alpha}(x')\neq\emptyset$. From a different view point, this is more of an advantage than a disadvantage because the validity guarantee~\eqref{eq:marginal-conformal-regression} requires $X_{n+1}$ from the same distribution as the $X$'s in the training data. Hence, an empty cross-section at $x$ as in Figure~\ref{fig:cross-section} actually informs the practitioner that the value $x$ is not likely, given the training values of $X$'s. This might be useful in raising a red flag when the practitioner is about to do extrapolation. For determining extrapolation, it might be better to first perform a dimension reduction that highlights the information in $X$ pertaining to predicting $Y$; this dimension reduction map can be based on the training split of the data.  

\vspace{0.1in}
\noindent\textbf{Conformal Regions with Focus on Response.}
The cross-sectional conformal prediction gives equal weight to the covariates and the response in that the prediction region $\widehat{\mathcal{R}}_{n,\alpha}$ does not focus more on either. In the regression or classification context, one might want to focus more on the response and ignore prediction for the covariates. Ignoring covariates for prediction means that the coverage guarantee for covariates can be 1, or in other words, the covariate cross-section of the region $\widehat{\mathcal{R}}_{n,\alpha}$ is the whole space of covariates. We will now provide ways to apply Theorem~\ref{thm:main-split-conformal} for prediction in regression and classification settings with the motivation above. This means that we will provide some concrete ways of designing $\widehat{f}_{\mathcal{T}}$ for the regression and classification framework.
\begin{itemize}[leftmargin=*]
\item \emph{Conditional Mean Estimation:} Based on the training data $\{(X_i, Y_i):1\le i\le n_1\}\subseteq\mathcal{X}\times\mathbb{R}$, estimate the conditional mean $\mathbb{E}[Y|X = x]$. Let the estimate be $\widehat{\mu}_{\mathcal{T}}(\cdot)$ based on the training data $\mathcal{T}$. Define, for $Z_i = (X_i, Y_i)$,
  \[
  \widehat{f}_{\mathcal{T}}(Z_i) ~:=~ |Y_i - \widehat{\mu}_{\mathcal{T}}(X_i)|,\quad n_1 + 1 \le i\le n+1.
  \]
  This yields a prediction region of the form $[\widehat{\mu}_{\mathcal{T}}(X_{n+1}) - t, \widehat{\mu}_{\mathcal{T}}(X_{n+1}) + t]$ for some $t$. The validity guarantee holds true irrespective of what $\widehat{\mu}_{\mathcal{T}}(\cdot)$ is. Any algorithm can be used to get an estimator; it may not even be consistent for $\mathbb{E}[Y|X = \cdot]$. Unlike the cross-sectional method, this function leads to a non-empty prediction region for all realized values of $X_{n+1}$ and as mentioned above, this can be a disadvantage. There exist many variations of {\clr these regions including those based on conditional variance normalization~\citep{lei2014distribution,lei2018distribution}, conditional quantile estimator~\citep{romano2019conformalized,kivaranovic2019adaptive,sesia2019comparison}, and conditional distributions~\citep{chernozhukov2019distributional,izbicki2019flexible,kuchibhotla2019nested}. These variations aim to get approximate conditional coverage and have width adapting to the conditional heteroscedasticity. See Corollary 1 of~\cite{sesia2019comparison} for details.}
\item \emph{Conditional Probability Density Estimation:} Although conditional validity~\eqref{eq:conditional-conformal-regression} is non-trivially impossible in a distribution-free setting, one can still construct a marginally valid prediction region that can asymptotically attain the conditional validity guarantee. Similar to the optimal prediction region~\eqref{eq:optimal-prediction-region} for the whole vector $Z_{n+1} = (X_{n+1}, Y_{n+1})$, the optimal conditional prediction region for $Y_{n+1}$ given $X_{n+1}$ is given by
  \begin{equation}\label{eq:optimal-conditional-conformal}
  {\mathcal{R}}_{\alpha}^{\mathrm{opt}}(x) ~:=~ \left\{y\in\mathcal{Y}:\,p(y|x) \ge t_{\alpha}(x)\right\},
  \end{equation}
  where $p(y|x)$ is the conditional probability density function of $Y$ given $X = x$ and $t_{\alpha}(x)$ is the largest $t\ge0$ such that
  \[
  \int_{y:p(y|x) \ge t} p(y|x)dy \ge 1 - \alpha.
  \]
  An imitation of this region is given by replacing $p(y|x)$ by an estimator $\widehat{p}_{\mathcal{T}}(y|x)$ based on the training data $(X_i, Y_i), 1\le i\le n_1$. This replacement does not guarantee any validity and any such guarantees depend on the accuracy of $\widehat{p}_{\mathcal{T}}(y|x)$ for $p(y|x)$. To guarantee validity, define for $z = (x, y)$,
  \begin{equation}\label{eq:optimal-score}
  \widehat{f}_{\mathcal{T}}(z) := 1 - \inf\{\alpha\in[0,1]:\,y\in\widehat{\mathcal{R}}^*_{\mathcal{T},\alpha}(x)\},
  \end{equation}
  where
  \[
  \widehat{\mathcal{R}}^*_{\mathcal{T},\alpha}(x) ~:=~ \left\{y\in\mathcal{Y}:\,\widehat{p}_{\mathcal{T}}(y|x) \ge \widehat{t}_{\mathcal{T},\alpha}(x)\right\}.
  \]
  Here $\widehat{t}_{\mathcal{T},\alpha}(x)$ is the largest $t>0$ such that $\int_{y:\widehat{p}_{\mathcal{T}}(y|x) \ge t}\widehat{p}_{\mathcal{T}}(y|x)dy \ge 1 - \alpha$. Now Theorem~\ref{thm:main-split-conformal} with $\widehat{f}_{\mathcal{T}}(\cdot)$ in~\eqref{eq:optimal-score} leads to a prediction region that is guaranteed to satisfy~\eqref{eq:marginal-conformal-regression}. We stress once again that although this region is imitating the optimal conditional prediction region~\eqref{eq:optimal-conditional-conformal}, it does \emph{not} have a finite sample conditional guarantee. Because of the imitation, it is expected that the region from Theorem~\ref{thm:main-split-conformal} based on~\eqref{eq:optimal-score} will asymptotically satisfy the conditional guarantee~\eqref{eq:conditional-conformal-regression}. {\clr Similar construction of conformal prediction sets for regression and classification has also been discussed in~\cite{izbicki2019flexible} and~\cite{kuchibhotla2019nested}.}
  Unlike the conditional mean estimation based prediction region, the region from~\eqref{eq:optimal-score} is sensible for classification and regression alike. {\clr For classification, the prediction set using $\widehat{f}_{\mathcal{T}}(\cdot)$ in~\eqref{eq:optimal-score} gathers those classes with higher estimated probabilities from the classifier.} See~\cite{romano2020classification} and~\cite{kuchibhotla2021nested} for a similar approach for classification. 
\end{itemize}
\subsection{Testing Interpretation of Conformal Prediction and other Variants}\label{sec:testing-interpre}
{\clr The full conformal method, as mentioned before, is computationally prohibitive, in general. But this method makes the full use of the data for prediction purposes. The split conformal method although computationally efficient uses some part of the data for training and only some part of the data for prediction calibration purposes. For this reason, many authors~\citep{balasubramanian2014conformal,lei2018distribution,barber2019predictive} have argued that split conformal method could incur statistical inefficiency due to this splitting. Several methods have been proposed to make better use of the data; see, for example,~\cite{carlsson2014aggregated,vovk2015cross,linusson2017calibration,vovkcombining,lei2018distribution,barber2019predictive,kuchibhotla2019nested,kim2020predictive}. In this section, we describe these variants briefly using the hypothesis testing interpretation of conformal prediction regions.}

{\clr Recall that the goal of conformal prediction is to construct a set $\widehat{\mathcal{R}}_{n,\alpha}$ based on $W_1, \ldots, W_n$ such that
\[
\mathbb{P}\left(W_{n+1}\in\widehat{\mathcal{R}}_{n,\alpha}\right) \ge 1 - \alpha.
\]
Following the duality of testing and confidence regions, we can formally think of testing the hypothesis $H_0:\,W_{n+1} = w$ for some value $w$. (This is not a traditional hypothesis because it relates to a random variable $W_{n+1}$.) Based on a test for $H_0$, a valid prediction set can be constructed by collecting all the $w$'s for which $H_0: W_{n+1} = w$ is not rejected. We will now define a $p$-value and the corresponding test when $W_1, \ldots, W_n, W_{n+1}$ are real-valued exchangeable random variables. The case of arbitrary spaces can be dealt with similarly using transformations. Define the $p$-value
\begin{equation}\label{eq:p-value-conformal}
P_{w} ~:=~ 1 - \frac{\mathrm{rank}(w; \{W_1, W_2, \ldots, W_n, w\}) + 1}{n+1}.
\end{equation}
From Corollary~\ref{cor:for-conformal-prediction}, it follows that $\mathbb{P}(P_{W_{n+1}} \le \alpha) \le \alpha$ for all $\alpha\in[0,1]$. In other words, $P_{w}$ is a valid $p$-value under $H_0: W_{n+1} = w$. Hence, the region $\{w:\,P_w > \alpha\}$ is a valid $(1-\alpha)$ prediction region. This $p$-value interpretation of conformal prediction method was mentioned in~\citet[Section 4.2]{shafer2008tutorial} and~\citet[Section 2]{MR3174619}, among others. For an interesting modification of these $p$-values in relation to conformal prediction, see~\cite{carlsson2015modifications}.

We are now ready to discuss the conformal prediction methods that lie in between the split and full conformal methods.
\begin{itemize}[leftmargin=*]
\item Jackknife and CV methods~\citep{vovk2015cross,barber2019predictive} are based on the idea of splitting the data into multiple disjoint folds (instead of just 2) and then combine the ranks or transformed variables in some way. To elaborate, we briefly describe the jackknife+ method from~\cite{barber2019predictive}. Suppose $Z_1, \ldots, Z_n, Z_{n+1}\in\mathcal{Z}$ are exchangeable random variables with the goal of predicting $Z_{n+1}$. Let $\widehat{f}_{-(i,j)}:\mathcal{Z}\to\mathbb{R}$ be a permutation invariant transformation computed based on $\{Z_1, \ldots, Z_{n+1}\}\setminus\{Z_i, Z_j\}$. In regression data with $Z_i = (X_i, Y_i)$, for example, $\widehat{f}_{-(i,j)}(z)$ can be $|y - \widehat{\mu}_{-(i,j)}(x)|$, where $\widehat{\mu}_{-(i,j)}(\cdot)$ is a regression function computed based on $\{Z_1, \ldots, Z_{n+1}\}\setminus\{Z_i, Z_j\}$. Define the prediction set for $Z_{n+1}$ as
\[
\widehat{\mathcal{R}}_{n,\alpha}^{\texttt{Jack}} ~:=~ \left\{z\in\mathcal{Z}:\,\sum_{j=1}^n \mathbbm{1}\left\{\widehat{f}_{-(n+1,i)}(z) > \widehat{f}_{-(i,n+1)}(Z_i)\right\} < (1-\alpha)(n+1)\right\}.
\]
Note that $\widehat{f}_{-(i,n+1)}(\cdot)\equiv\widehat{f}_{-(n+1,i)}(\cdot)$ and that these transformations can be computed without the knowledge of $Z_{n+1}$. Further, $\widehat{f}_{-(i,n+1)}(\cdot)$ is a leave-one-out transformation on the data $Z_1, \ldots, Z_n$. Theorem~1 of~\cite{barber2019predictive} can be used to prove that $\mathbb{P}(Z_{n+1}\in\widehat{\mathcal{R}}_{n,\alpha}^{\texttt{Jack}}) \ge 1-2\alpha$, although the theorem is only stated for regression data with absolute residual. The proof of Theorem~1 of~\cite{barber2019predictive} hinges on the fact that
\begin{equation}\label{eq:transformation-jackknife+}
G\begin{pmatrix}Z_1\\\vdots\\Z_{n+1}\end{pmatrix} := \begin{pmatrix}W_1\\\vdots\\W_{n+1}\end{pmatrix} = \begin{pmatrix}\sum_{j=2}^{n+1}\mathbbm{1}\{\widehat{f}_{-(1,j)}(Z_1) > \widehat{f}_{-(j,1)}(Z_j)\}\\\vdots\\
\sum_{j=1}^{n}\mathbbm{1}\{\widehat{f}_{-(n+1,j)}(Z_{n+1}) > \widehat{f}_{-(j,n+1)}(Z_{j})\}\end{pmatrix},
\end{equation}
is an exchangeability preserving transformation; this can be verified using Theorem~\ref{thm:transformations-exchangeability}.  Step 1 in Section 6 of~\cite{barber2019predictive} shows that the number of coordinates in the right hand side of~\eqref{eq:transformation-jackknife+} that are larger than or equal to $(1-\alpha)(n+1)$ is bounded by $2\alpha(n+1)$; this is a deterministic inequality and does not require exchangeability of $Z_1, \ldots, Z_{n+1}$. Using the fact that $G(\cdot)$ is exchangeability preserving, we get that
\begin{align*}
\mathbb{P}\left(W_{n+1} \ge (1-\alpha)(n+1)\right) &= \frac{1}{n+1}\sum_{j=1}^{n+1}\mathbb{P}\left(W_j\ge(1-\alpha)(n+1)\right)\\ &= \mathbb{E}\left[\frac{\sum_{j=1}^n \mathbbm{1}\{W_j \ge (1-\alpha)(n+1)\}}{n+1}\right] \le 2\alpha.
\end{align*}
Here the inequality above follows from Step 1 in Section 6 of~\cite{barber2019predictive}. The CV+ method is defined similarly where instead of leave-one-out, one uses a leave-a-fold-out; see Section 3 of~\cite{barber2019predictive} for details. Also, see~\cite{solari2021multi} for a different argument.
\item Subsampling or repeated split methods~\citep{carlsson2014aggregated,lei2018distribution,kuchibhotla2019nested} repeat the split conformal method several times on the data and combine the resulting prediction sets in some way. To elaborate, we briefly discuss the subsampling or Bonferroni method discussed in~\cite{lei2018distribution} and~\cite{kuchibhotla2019nested}. Recall that the split conformal method can be interpreted in terms of a retention region from a $p$-value~\eqref{eq:p-value-conformal}. If we repeat the splitting process on the data $K$ times, then we get $K$ $p$-values $P_{Z_{n+1}}^{(1)}, \ldots, P_{Z_{n+1}}^{(K)}$. It is very important to observe that these are dependent $p$-values, dependent through $Z_{n+1}$. This implies that $K\min_{1\le k\le K}P_{Z_{n+1}}^{(k)}$ is also a valid $p$-value:
\[
\mathbb{P}\left(K\min_{1\le k\le K}P_{Z_{n+1}}^{(k)} \le \alpha\right) \le \mathbb{P}\left(\bigcup_{k=1}^K\left\{P_{Z_{n+1}}^{(k)} \le \frac{\alpha}{K}\right\}\right) \le \sum_{k=1}^K \frac{\alpha}{K} \le \alpha.
\] 
Hence, we get that $\{z:\,P_z^{(k)} > \alpha/K\mbox{ for all }1\le k\le K\}$ is a valid $(1-\alpha)$ prediction region for $Z_{n+1}$. The combination of $p$-values above is the Bonferroni correction from the multiple testing literature. One can use other combinations of $p$-values such as twice the arithmetic or geometric mean, and so on; see~\cite{vovkcombining} for more examples. The use of multiple testing for other conformal methods can be seen in~\citet[Section 2.3]{lei2018distribution},~\citet[Section 4]{vovkcombining},~\citet[Appendix B \& C]{kuchibhotla2019nested}.
\end{itemize} 
Although these variants make better use of the full data, there is a clear lack of great advantage of these complicated methods in performance per computational cost in comparison to the split conformal method. Firstly, these variants require computing the transformation multiple times. Secondly, they can be both conservative and anti-conservative in practice. Finally, the volume of the resulting prediction regions can also be larger than that of the split method. See~\citet[Figure 2]{barber2019predictive} and~\citet[Table 2]{kuchibhotla2019nested} for some comparisons. This is not to say that split conformal is always the best. Conformal methods that make use of (close to) full data perform best when the dimension of the training algorithm is close to the sample size; for example, fitting linear regression with $n$ observations and $d ~(\approx n)$ covariates~\citep[Figure 2]{barber2019predictive}. 
With machine learning algorithms such as random forests that automatically yield several training and calibration sets, the out-of-bag or aggregated conformal methods~\citep{kim2020predictive,kuchibhotla2019nested} can yield better performance in comparison to split conformal without increasing the computational cost. 
}
\section{Nonparmetric Rank Tests}\label{SEC:RANK-TESTS}
Testing equality of distributions and independence of random vectors are two of the most fundamental problems in statistics. In the following two sections, we discuss each of these problems and detail the implications of exchangeability. We have seen in the previous section that a basic prediction interval for real-valued random variables can be used to construct prediction sets for random variables in arbitrary spaces by a data-driven transformation. In this section, we show that the classical non-parametric rank tests, defined for real-valued random variables, can also be used for tests for random variables in arbitrary spaces by a data-driven transformation. All this is made possible by exchangeability and its consequences. {\clr It should be mentioned that this dimension reduction idea in rank tests is not new and has been discussed in some works such as~\cite{matthews1996nonparametric} and~\cite{friedman2003multivariate}.}

A brief description of the tests is as follows:
\begin{enumerate}
  \item \textbf{Equality of Distributions:} Given two datasets, one might want to test if the datasets are obtained from the same distribution. More formally, given $n$ i.i.d. observations {\clr $X_1, \ldots, X_n$} from $P_X$ and $m$ i.i.d. observations {\clr $Y_1, \ldots, Y_m$} from $P_Y$ with both sets of observations independent, we want to test
  \begin{equation}\label{eq:null-equality-of-dist}
  H_0:\,P_X = P_Y\quad\mbox{versus}\quad H_1: P_X\neq P_Y.
  \end{equation}
  This is also known as a two-sample testing problem and has numerous applications in pharmaceutical studies~\citep{farris1999between}, causal inference~\citep{folkes1987field}, remote sensing~\citep{conradsen2003test}, and econometrics~\citep{mayer1975selecting}.
  \item \textbf{Independence:} Given $n$ i.i.d. paired observations $(X_i, Y_i), 1\le i\le n$, one might want to test if $X_i$ is independent of $Y_i$. Formally, {\clr suppose} the distribution of $(X_i, Y_i)$ is $P$, the {\clr marginal} distribution of $X_i$ is $P_X$, and the {\clr marginal} distribution of $Y_i$ is $P_Y$. Then we want to test
  \[
  H_0:\,P = P_X\otimes P_Y\quad\mbox{versus}\quad H_1: P \neq P_X\otimes P_Y.
  \]
  Independence testing has found applications in statistical genetics~\citep{liu2010versatile}, marketing and finance~\citep{grover1985probabilistic}, survival analysis~\citep{martin2005testing}, and ecological risk
assessment~\citep{dishion1999middle}. Furthermore, the test for equality of distributions can be formulated as a test for independence {\clr by defining $Y$ as a binary random variable labeling the sample to which $X$ belongs to; see, e.g.,~\citet[Section 1]{heller2016consistent}.}
\end{enumerate}
{\clr Unlike in the case of conformal prediction, in this section, we will assume that the underlying observations are independent and identically distributed. The reason for this change in assumption can be understood as follows. Consider the testing problem~\eqref{eq:null-equality-of-dist} of equality of distributions. Under the null hypothesis $H_0$, we want the data obtained by combining $\{X_1, \ldots, X_n\}$ and $\{Y_1, \ldots, Y_m\}$ to form an exchangeable sequence. For this exchangeability, under the independence of $\{X_1, \ldots, X_n\}$ and $\{Y_1, \ldots, Y_m\}$, the i.i.d. assumption seems the most sensible.}
\subsection{Testing Equality of Distributions}\label{sec:equality-of-distributions}
Suppose $X_1, \ldots, X_n$ are independent and identically distributed random variables from a probability measure $P_X$ and $Y_1, \ldots, Y_m$ are independent and identically distributed random variables from a probability measure $P_Y$. Random variables $X_1, \ldots, X_n$ are independent of $Y_1, \ldots, Y_m$. The hypothesis to test is
\begin{equation}\label{eq:equality-of-distributions}
H_0:\,P_X = P_Y\quad\mbox{versus}\quad H_1:\,P_X\neq P_Y.
\end{equation}
There exist numerous tests for this hypothesis. We refer to reader to~\cite{bhattacharya2015two,bhattacharya2019general} and~\cite{deb2019multivariate} for an overview of the existing literature. Wilcoxon rank sum test~\citep[Section 4.1]{hollander2013nonparametric} is one of the classical rank tests for this hypothesis when the observations are real-valued. Let $\mathcal{C}_{n+m} := \{Z_1, \ldots, Z_{n+m}\}$ represent the random variables $X_1, \ldots, X_n$ and $Y_1, \ldots, Y_m$ put together. Under $H_0$, the collection of random variables in $\mathcal{C}_{n+m}$ are independent and identically distributed and hence, by Theorem~\ref{thm:rank-distribution}, $\mathrm{rank}(Z_1; \mathcal{C}_{n+m}), \ldots, \mathrm{rank}(Z_{n+m};\mathcal{C}_{n+m})$ are distributed uniformly over all permutations of $[n+m]$; recall $[n+m]=\{1,2,\ldots,n+m\}$. The Wilcoxon rank sum test statistic is given by
\[
T_n ~:=~ \sum_{i=1}^n \mathrm{rank}(Z_i;\,\mathcal{C}_{n+m}),
\]  
the sum of ranks of $X_1, \ldots, X_n$ among $\mathcal{C}_{n+m}$. Recall the definition of rank from Definition~\ref{def:rank-definition}. From this definition, it follows that the distribution of $T_n$ does not depend on the true distributions $P_X$ and $P_Y$ under $H_0$. 

Using this, we now extend the Wilcoxon test to random variables taking values in arbitrary space $\mathcal{Z}$. Let $\widehat{f}_{n+m}:\mathcal{Z}\to\mathbb{R}$ be any transformation that depends \emph{permutation invariantly} on $Z_1, \ldots, Z_{n+m}\in\mathcal{C}_{n+m}$. This means that if we write
\[
\widehat{f}_{n+m}(z) ~=~ \widehat{f}_{n+m}(z;\,Z_1, \ldots, Z_{n+m}),
\] 
then permutation invariance means for any $z\in\mathcal{Z}$ and any $\pi:[n+m]\to[n+m]$, 
\[
\widehat{f}_{n+m}(z;\, Z_1, \ldots, Z_{n+m}) ~=~ \widehat{f}_{n+m}(z;\,Z_{\pi(1)}, \ldots, Z_{\pi(n+m)}).
\]
Proposition~\ref{prop:permutation-invariant-exchangeable} (in Appendix~\ref{appsec:auxiliary-results}) shows that $W_1 := \widehat{f}_{n+m}(Z_1), \ldots, W_{n+m} := \widehat{f}_{n+m}(Z_{n+m})$ are exchangeable and hence
\[
T_n^{\mathrm{exch}} := \sum_{i=1}^n \mathrm{rank}(W_i;\, \{W_1, \ldots, W_{n+m}\}),
\]
the sum of ranks of $W_1, \ldots, W_n$ among the collection $\{W_1, \ldots, W_{n+m}\}$, is also distribution-free and has the same distribution as $T_n$ under $H_0$. This proves the following result.
\begin{thm}\label{thm:wilcoxon-exchangeable}
Suppose $X_1, \ldots, X_n$ are independent and identically distributed random variables from a probability measure $P_X$ supported on $\mathcal{Z}$ and $Y_1, \ldots, Y_m$ are independent and identically distributed random variables from a probability measure $P_Y$ supported also on $\mathcal{Z}$. Then for any transformation $\widehat{f}_{n+m}$ depending permutation invariantly on $X_1, \ldots, X_n, Y_1, \ldots, Y_m$,
\begin{equation}\label{eq:T-n-exchangeable}
T_n^{\mathrm{exch}} ~:=~ \sum_{i=1}^n \mathrm{rank}(\widehat{f}_{n+m}(X_i);\,\{\widehat{f}_{n+m}(X_1), \ldots, \widehat{f}_{n+m}(Y_m)\}),
\end{equation}
is distribution-free under $H_0$ and matches the null distribution of the Wilcoxon rank sum statistic. 
\end{thm}
The main conclusion of the theorem is that any permutation invariant transformation $\widehat{f}_{n+m}$ obtained from the data will lead to a distribution-free test with a finite sample control of the type I error, irrespective of the domain of the data. This is the analogy with {\clr  the full} conformal prediction {\clr method} where for any {\clr permutation invariant} transformation~\eqref{eq:permutation-invariance-full-conformal}, the prediction region has a finite sample control of the coverage. 
{\clr Theorem~\ref{thm:wilcoxon-exchangeable} is obvious if $\widehat{f}_{n+m}$ does not depend on the data $Z_1, \ldots, Z_{n+m}$. Such application of rank tests after dimension reduction is well-known. The main novelty of Theorem~\ref{thm:wilcoxon-exchangeable} is that the transformation $\widehat{f}_{n+m}$ can depend on the data. Although relatively easy to derive based on Theorem~\ref{thm:transformations-exchangeability}, Theorem~\ref{thm:wilcoxon-exchangeable} is new to the best of our knowledge.}
\begin{tcolorbox}
\noindent\textbf{Pseudocode 3:} Computationally, Theorem~\ref{thm:wilcoxon-exchangeable} works as follows:
\begin{enumerate}
    \item Rename the data as $$Z_1 = X_1, \ldots, Z_n = X_n, Z_{n+1} = Y_1, \ldots, Z_{n+m} = Y_m.$$
    \item Find a permutation invariant transformation $\widehat{f}_{n+m}$ based on $Z_1, \ldots, Z_{n+m}$. This can be any cluster revealing transformation; see Section~\ref{subsec:examples-equality-of-dist} below.
    \item Take $\xi = 10^{-8}$ and compute the ranks of the first $n$ elements $\widehat{f}_{n+m}(Z_1), \ldots, \widehat{f}_{n+m}(Z_n)$ among $\widehat{f}_{n+m}(Z_{1}), \ldots, \widehat{f}_{n+m}(Z_{n+m})$.
    \item Sum the ranks and compute $T_n^{\mathrm{exch}}$. Reject the null hypothesis $H_0$ if $T_n^{\mathrm{exch}}$ deviates from its expected value~$n(n+m+1)/2$. The critical value here is the same as that of Wilcoxon rank sum test; see~\citet[Section 4.1]{hollander2013nonparametric} and~\cite{wilcoxon1945individual} for details.
\end{enumerate}
\end{tcolorbox}
\subsubsection{Some Concrete Examples.}\label{subsec:examples-equality-of-dist} In the following, we discuss a few examples of $\widehat{f}_{n+m}$ in Theorem~\ref{thm:wilcoxon-exchangeable}. The most generic application is based on applying any of the many unsupervised clustering algorithms.
\begin{enumerate}[leftmargin=*]
    \item \textbf{$k$-means Clustering:} Suppose $Z_1, \ldots, Z_{n+m}$ are elements of a space $\mathcal{Z}$ with a quasi-norm $\|\cdot\|$.\footnote{Quasi-norm means that $\|\cdot\|$ does not need to satisfy the triangle inequality.} Fix $k\ge1$, the number of clusters and apply the $k$-means clustering algorithm on the data $Z_1, \ldots, Z_{n+m}$, that is, find $\widehat{c}_1, \ldots, \widehat{c}_k\in\mathcal{Z}$ such that
    \begin{equation}\label{eq:k-means-proper}
    (\widehat{c}_1, \ldots, \widehat{c}_k) ~:=~ \argmin_{c_1, \ldots, c_k\in\mathcal{Z}}\,\sum_{i=1}^{n+m}\,\min_{1\le j\le k}\|Z_i - c_j\|^2.
    \end{equation}
    {\clr Note that the minimizer can at best be unique up to permutations, i.e., if $(\widehat{c}_1, \ldots, \widehat{c}_k)$ is a minimizer, then $(\widehat{c}_{\pi(1)}, \ldots, \widehat{c}_{\pi(k)})$ is also a minimizer for arbitrary permutation $\pi:[k]\to[k]$. For our purposes, we choose (arbitrarily) and fix a minimizer.}
    It is easy to verify that $(\widehat{c}_1, \ldots, \widehat{c}_k)$ is a permutation invariant function of $\{Z_1, \ldots, Z_{n+m}\}$. 
    One transformation $\widehat{f}_{n+m}$ based on this clustering method is given by
    \[
    \widehat{f}_{n+m}(Z_i) ~:=~ \|Z_i - \widehat{c}_1\| - \min_{1\le j\le k}\|Z_i - \widehat{c}_j\|. 
    \] 
    To gain intuition for this particular permutation invariant transformation, suppose $H_1$ is true (that is, $P_X\neq P_Y$) and $P_X = N(\mu_X, \sigma^2)$ and $P_Y = N(\mu_Y, \sigma^2)$. If we perform $k$-means clustering with $k=2$, then asymptotically $\widehat{c}_1 = \mu_X$ and $\widehat{c}_2 = \mu_Y$ (up to labels) and hence, asymptotically,
    \begin{align*}
     &\min_{1\le j\le k}\|Z_i - \widehat{c}_j\| = \|Z_i - \widehat{c}_1\|,\quad\Rightarrow\quad \widehat{f}_{n+m}(Z_i) = 0,\quad\mbox{for}\quad 1\le i\le m,\\
     &\min_{1\le j\le k}\|Z_i - \widehat{c}_j\| = \|Z_i - \widehat{c}_2\|,\quad\mbox{for}\quad m + 1 \le i\le m+n,\\
     &\quad\Rightarrow\; \widehat{f}_{n+m}(Z_i) = \|Z_i - \widehat{c}_1\| - \|Z_i - \widehat{c}_2\| > 0,\quad\mbox{for}\quad i > m.
    \end{align*}
    This implies that under $H_1$ (asymptotically) $T_n^{\mathrm{exch}} = \sum_{i=1}^n i = n(n+1)/2$ which is significantly smaller than the mean leading to a rejection.

    For the purposes of Theorem~\ref{thm:wilcoxon-exchangeable}, it does not matter if one is able to obtain the global minimum in~\eqref{eq:k-means-proper}. Any method for obtaining $\widehat{c}_1, \ldots, \widehat{c}_k$ is permissible, as long as the procedure does not depend on the $X, Y$ labels. For instance, Lloyd's $k$-means clustering algorithm, $k$-means++ algorithm~\citep{arthur2007k} and many other variants~\citep{celebi2013comparative,hamerly2010making,ding2015yinyang,hamerly2015accelerating,shen2017compressed,newling2016fast} will work. In addition to the variants of the algorithms, one can also use the data to choose $k$, the number of clusters based on the data. In particular, the well-known elbow rule can be used for this purpose. 
    \item \textbf{Dimension Reduction and $k$-means:} In the same setting as above, due to the high computational complexity~\citep{gronlund2017fast,aloise2009np} of the $k$-means clustering, one might opt to perform a preliminary dimension reduction and then use $k$-means clustering. {\clr Another reason for dimension reduction could be a possibility that distributions differ along a low-dimensional projection.} Many of the dimension reduction techniques discussed in Section~\ref{subsec:examples-conformal} can be used. The final transformation $\widehat{f}_{n+m}$ is given by
    \[
    \widehat{f}_{n+m}(Z_i) ~:=~ \|\Pi(Z_i) - \widehat{c}_1\| ~-~ \min_{1\le j\le k}\|\Pi(Z_i) - \widehat{c}_{j}\|,
    \]   
    where $\Pi(\cdot)$ denotes the preliminary dimension reduction map with $\Pi(Z_1)$, $\ldots$, $\Pi(Z_{n+m})$ representing the dimension reduced data and $\widehat{c}_1, \ldots, \widehat{c}_k$ now represent the $k$-cluster centers obtained from the dimension reduced data.
    \item \textbf{Clustering and Density Estimation:} In this case, we provide a generic method of converting any unsupervised clustering method into a valid transformation for Theorem~\ref{thm:wilcoxon-exchangeable}. Suppose we have an unsupervised clustering procedure $\mathcal{P}$ that partitions the collection of random variables $Z_1, \ldots, Z_{n+m}$ into disjoint sets $B_1, \ldots, B_k$ (where $k$ could itself be a part of $\mathcal{P}$). {\clr This clustering procedure should treat the data $Z_1, \ldots, Z_{n+m}$ in a permutation invariant way.} See~\cite{wasserman-notes} for some examples. From all the random variables in $B_j, 1\le j\le k$, estimate the density in a permutation invariant way; see~\cite{wang2019nonparametric} and~\cite{kim2018uniform} for some examples. 
    Then the transformation $\widehat{f}_{n+m}(\cdot)$ is given by
    \[
    \widehat{f}_{n+m}(Z_i) ~:=~ \frac{\widehat{p}_{B_1}(Z_i)}{\max_{1\le j\le k}\,\widehat{p}_{B_j}(Z_i)}.
    \]
    To gain intuition for this transformation, observe that if we obtain $\widehat{c}_1, \ldots, \widehat{c}_k$ from a nearest neighbor clustering and we take the density estimator $\widehat{p}_{B_j}(x) = \phi(x; \widehat{c}_j, I),$ the multivariate normal density with location $\widehat{c}_j$ and variance identity, then $\widehat{f}_{n+m}(Z_i)$ is a simple transformation of $\|Z_i - \widehat{c}_1\| - \min_{1\le j\le k}\|Z_i - \widehat{c}_j\|$ matching the one in first example above. 

    As in the previous example, one can first apply a dimension reduction procedure on the data and then apply the unsupervised clustering method. In this case, the density estimator could be based on the dimension reduced data.
\end{enumerate} 
\subsection{Testing Independence}\label{sec:test-independence}
Suppose $(X_1, Y_1), \ldots, (X_n, Y_n)\in\mathcal{X}\times\mathcal{Y}$ are independent and identically distributed random variables from a probability measure $P$ on $\mathcal{X}\times\mathcal{Y}$. If the marginal distribution of $X_i$'s is $P_X$ and the marginal distribution of $Y_i$'s is $P_Y$, then the hypothesis to test is
\begin{equation}\label{eq:independence}
H_0:\, P = P_X\otimes P_Y\quad\mbox{versus}\quad H_1:\, P \neq P_X\otimes P_Y.
\end{equation}
Here $P_X\otimes P_Y$ represents the joint probability distribution with independent marginals of $P_X$ and $P_Y$. Similar to the problem of testing equality of distributions, there exist numerous tests for~\eqref{eq:independence} and we refer to~\cite{han2017distribution,deb2019multivariate,shi2020distribution} for an overview. One of the classical tests for independence hypothesis is based on the Spearman's rank correlation~\citep[Section 8.5]{hollander2013nonparametric}. This test applies when $P_X$ and $P_Y$ are supported on the real line. Suppose the observations are $(X_1, Y_1), \ldots, (X_n, Y_n)\in\mathbb{R}^2$. Let $R_1^X, \ldots, R_n^X$ denote the ranks of $X_1, \ldots, X_n$ and $R_1^{Y}, \ldots, R_n^{Y}$ denote the ranks of $Y_1, \ldots, Y_n$. Under the null hypothesis $H_0$, $(R_1^X, \ldots, R_n^X)$ and $(R_1^Y, \ldots, R_n^Y)$ are independent random vectors. Further each of these vectors is distributed as uniform on all permutations of $(1,2,\ldots, n)$ because of Theorem~\ref{thm:rank-distribution}. The Spearman's rank correlation is given by
\begin{equation}\label{eq:spearman-rank-corr}
\rho_n ~:=~ 1 - \frac{6}{n}\sum_{i=1}^n\frac{(R_i^X - R_i^Y)^2}{(n^2 - 1)}. 
\end{equation}
Because the distributions of $(R_1^X, \ldots, R_n^X)$ and $(R_1^Y, \ldots, R_n^Y)$ do not depend on $P_X$ and $P_Y$, $\rho_n$ in~\eqref{eq:spearman-rank-corr} has the same distribution (under $H_0$) irrespective of what $P_X$ and $P_Y$ are. 

Noting that the distribution-free nature of the test depends only on the fact that ranks are distribution-free, we get that we can transform the data in each coordinate almost arbitrarily. Let $\widehat{f}_{X}$ be a transformation that depends on $X_1, \ldots, X_n$ permutation invariantly and let $\widehat{f}_{Y}$ be a transformation that depends on $Y_1, \ldots, Y_n$ permutation invariantly. Then by Theorem~\ref{thm:transformations-exchangeability}, $\widehat{f}_X(X_1), \ldots, \widehat{f}_X(X_n)$ are exchangeable and $\widehat{f}_Y(Y_1), \ldots, \widehat{f}_Y(Y_n)$ are also exchangeable. Further under the null hypothesis $H_0$, the vectors $(\widehat{f}_X(X_1), \ldots, \widehat{f}_X(X_n))$ and $(\widehat{f}_Y(Y_1), \ldots, \widehat{f}_Y(Y_n))$ are independent. This leads to the following result.
\begin{thm}\label{thm:independence}
Suppose $(X_1, Y_1), \ldots, (X_n, Y_n)\in\mathcal{X}\times\mathcal{Y}$ are independent observations. Suppose $\widehat{f}_X$ and $\widehat{f}_Y$ are transformations depending permutation invariantly only on $(X_1, \ldots, X_n)$ and $(Y_1, \ldots, Y_n)$, respectively. Then the statistic
\[
\rho_n(\widehat{f}_X, \widehat{f}_Y) ~:=~ 1 - \frac{6}{n}\sum_{i=1}^n \frac{(R_i^X - R_i^Y)^2}{n^2 - 1},
\]
has the same distribution as Spearman's rank correlation~\eqref{eq:spearman-rank-corr} under $H_0$. Here 
\begin{align*}
R_i^X ~&:=~ \mathrm{rank}(\widehat{f}_X(X_i);\, \{\widehat{f}_X(X_1), \ldots, \widehat{f}_X(X_n)\}),\\ 
R_i^Y ~&:=~ \mathrm{rank}(\widehat{f}_Y(Y_i);\,\{\widehat{f}_Y(X_1), \ldots, \widehat{f}_Y(Y_n)\}).
\end{align*}
\end{thm} 
The main conclusion of Theorem~\ref{thm:independence} is that any permutation invariant data-driven transformations $\widehat{f}_X$ and $\widehat{f}_Y$ will lead to a distribution-free finite sample valid test for the independence hypothesis~\eqref{eq:independence}, irrespective of the domain of the data. Theorem~\ref{thm:independence} is, however, lacking in one important way. If allowed, one might want to use transformations $\widehat{f}_X$ and $\widehat{f}_Y$ depending on $(X_i, Y_i), i\in[n]$ that leads to the maximal correlation between $\widehat{f}_X(X_i), i\in [n]$ and $\widehat{f}_Y(Y_i), i\in[n]$~\citep{renyi1959measures}. These transformations would depend on both coordinates. This, however, does not lead to validity at least through Theorem~\ref{thm:independence}. Once again, Theorem~\ref{thm:independence} can be seen as an analogy of the full conformal prediction method.
\subsection{Rank Tests based on Sample Splitting}
In the previous sections, we have restricted the nature of data-driven transformations; in case of equality of distributions, the transformations should not depend on the $X, Y$ labels and in case of independence, the transformations have to depend on $X_1, \ldots, X_n$ and $Y_1, \ldots, Y_n$ only marginally not jointly. Although the tests control the type I error, these restrictions can drastically effect the power. We can avoid these restrictions based on sample splitting, which will be described below. {\clr The following procedures based on sample splitting can be thought of as analogues to the split conformal prediction method.}

{\clr Although the full sample based transformation is lacking in the literature, the sample splitting transformation for rank tests has been described in the literature albeit seems not widely known. See~\cite{friedman2003multivariate} and~\cite{vayatis2009auc}.}

\vspace{0.1in}
\noindent\textbf{Equality of Distributions:} Recall that the hypotheses are $H_0: P_X = P_Y$ and $H_1: P_X\neq P_Y$. The observations are $X_1, \ldots, X_n$ i.i.d from $P_X$ and $Y_1, \ldots, Y_m$ i.i.d. from $P_Y$. Under the null hypothesis $H_0$, $Z_1, \ldots, Z_{n+m}$ defined by {\clr $Z_i = X_i$ for $1\le i\le n$ and $Z_i = Y_{i-n}$ for $n+1 \le i\le n+m$} are independent and identically distributed. Randomly split this collection into two parts:
\[
S_1 ~:=~ \{Z_i:\,i\in\mathcal{I}_1\}\quad\mbox{and}\quad S_2 ~:=~ \{Z_i:\,i\in\mathcal{I}_2\},
\]
where $\mathcal{I}_1$ contains a random subset of $[n+m]$ and $\mathcal{I}_2 = [n+m]\setminus\mathcal{I}_1$. For any transformation $\widehat{f}_{S_1}$ based on $S_1$, Proposition~\ref{prop:split-exchangeable} proves that $\widehat{f}_{S_1}(Z_i), i\in\mathcal{I}_2 = \mathcal{I}_1^c$ are exchangeable. Hence, 
\[
T^{\mathrm{split}} ~:=~ \sum_{i = 1, i\in\mathcal{I}_2}^n \mathrm{rank}(\widehat{f}_{S_1}(Z_i);\;\{\widehat{f}_{S_1}(Z_j):\,j\in \mathcal{I}_2\}),
\]
the sum of ranks of the $X_i$ random variables in the second split $S_2$ is a valid test statistic. It follows from the discussion in Section~\ref{sec:equality-of-distributions} that under $H_0$, $T^{\mathrm{split}}$ has a distribution independent of $P_X = P_Y$ and the critical values can be obtained from the Wilcoxon rank sum test based on the total sample size of $|\mathcal{I}_2|$.

In comparison to the test statistic $T_n^{\mathrm{exch}}$ in Section~\ref{sec:equality-of-distributions}, the transformation $\widehat{f}_{S_1}(\cdot)$ can now depend arbitrarily on the $X, Y$ labels. In particular, we describe a specific example below. Based on the first split $S_1$ of the data, obtain $\widehat{p}_{X,S_1}(\cdot)$ and $\widehat{p}_{Y, S_1}(\cdot)$, the estimators of density of $P_X$ and $P_Y$; the density estimator can be arbitrary~\citep{kim2018uniform,wang2019nonparametric}. The final transformation $\widehat{f}_{S_1}(\cdot)$ is given by
\[
\widehat{f}_{S_1}(Z_i) ~:=~ \frac{\widehat{p}_{X,S_1}(Z_i)}{\widehat{p}_{Y, S_1}(Z_i)},\quad\mbox{for all}\quad i\in \mathcal{I}_2.
\]
Further examples can be derived by writing the data in $S_1$ as $(Z_i, W_i), i\in\mathcal{I}_1$, where $W_i = 0$ if $Z_i\sim P_X$ and $W_i = 1$ if $Z_i\sim P_Y$ and finding an estimator $\widehat{\eta}_{S_1}(\cdot)$ of the conditional probability $\eta(z) := \mathbb{P}(W = 1|Z = z)$ based only on $S_1$. The final transformation then would be $\widehat{f}_{S_1}(Z_i) = \widehat{\eta}_{S_1}(Z_i)$, which would naturally be higher for cases where $W_i = 1$ for $i\in\mathcal{I}_2$ (under $H_1$). This can also be combined with the methods of central subspace estimation~\citep{ma2013efficient}, which is fruitful when $\eta(\cdot)$ depends on a few coordinates or directions. For instance, if $P_X$ and $P_Y$ only differ in the distribution of the first coordinate, then one can at first apply a subspace estimation algorithm for $W$ on $Z$ with the data in $S_1$ and then find a classifier based on the reduced subspace.

\vspace{0.1in}
\noindent\textbf{Independence:} Recall that the hypotheses are $H_0: P = P_{X}\otimes P_Y$ and $H_1: P\neq P_X\otimes P_Y$. The observations are $(X_1, Y_1), \ldots, (X_n, Y_n)$ which are i.i.d. from $P$. As before, randomly split the data into two parts:
\[
S_1 := \{(X_i, Y_i):\,i\in\mathcal{I}_1\}\quad\mbox{and}\quad S_2 := \{(X_i, Y_i):\,i\in\mathcal{I}_2\},
\] 
where $\mathcal{I}_1$ contains a random subset of $\{1,2,\ldots,n\}$ and $\mathcal{I}_2 = \{1,2,\ldots,n\}\setminus\mathcal{I}_1$. For any transformations $\widehat{f}_{X, S_1}(\cdot)$ and $\widehat{f}_{Y, S_1}(\cdot)$ based on $S_1$, Proposition~\ref{prop:split-exchangeable} yields that the bivariate random vectors $(\widehat{f}_{X,S_1}(X_i), \widehat{f}_{Y, S_1}(Y_i)), i\in\mathcal{I}_2$ are exchangeable and hence
\[
\rho^{\mathrm{split}}(\widehat{f}_{X, S_1}, \widehat{f}_{Y, S_1}) ~:=~ 1 - \frac{6}{|\mathcal{I}_2|}\sum_{i\in\mathcal{I}_2}\frac{(R_i^X - R_i^Y)^2}{|\mathcal{I}_2|^2 - 1},
\]
is a valid test statistic, where
\begin{align*}
R_i^X ~&:=~ \mathrm{rank}(\widehat{f}_{X, S_1}(X_i);\,\{\widehat{f}_{X,S_1}(X_j):\,j\in\mathcal{I}_2\}),\\
R_i^Y ~&:=~ \mathrm{rank}(\widehat{f}_{Y, S_1}(Y_i);\,\{\widehat{f}_{Y, S_1}(Y_j):\,j\in\mathcal{I}_2\}).
\end{align*}
Following the discussion in Section~\ref{sec:test-independence}, we conclude that under the null hypothesis $H_0:P = P_X\otimes P_Y$, $\rho^{\mathrm{split}}(\widehat{f}_{X, S_1}, \widehat{f}_{Y, S_1})$ has the same distribution as the Spearman's rank correlation based on sample size $|\mathcal{I}_2|$.

In comparison to the test statistic in Section~\ref{sec:test-independence}, we can now use transformations $\widehat{f}_{X,S_1}$ and $\widehat{f}_{Y, S_1}$ that can depend on $(X_i, Y_i), i\in S_1$ jointly not just marginally. We describe a specific example below. Based on the first split $S_1$ of the data, obtain transformations $\widehat{f}_{X, S_1}(\cdot)$ and $\widehat{f}_{Y, S_1}(\cdot)$ that maximize the ``correlation'' between $X_i, Y_i$ for $i\in S_1$; any technique can be used here and the ``correlation'' measure is also arbitrary. See~\cite{breiman1985estimating} for an example and one can also mix this methodology with dimension reduction techniques~\citep{ma2013efficient}. These transformations can be used in the statistic $\rho^{\mathrm{split}}(\widehat{f}_{X, S_1}, \widehat{f}_{Y, S_1})$ above. The critical values for this statistic can be obtained as before.

Summarizing the discussion in Section~\ref{SEC:RANK-TESTS}, we have shown that the distribution-free nature of the rank tests continues to hold under a large class of data-driven transformations. In all these sections, we have described the procedures only through two classical tests: Wilcoxon rank-sum test and Spearman's rank correlation test. Because most rank tests only depend on the fact that ranks are distributed uniformly over all permutations, the procedures can also be used with other rank tests~\citep{hollander2013nonparametric}.
\section{Summary and Concluding Remarks}\label{sec:summary}
We have described the fundamental concept of exchangeability and its implications for prediction regions as well as rank tests. By describing the basic components, the intention is to bring the conformal prediction more into practice and also to show the wide range of flexibility hiding within the rank tests. In both these topics, we have (intentionally) not done an in-depth survey of the existing literature. {\clr We encourage the reader to refer to the cited literature to explore these topics further. } 

Of course, in both cases (prediction and testing), it is also of interest to understand the ``power''. For a prediction region, this could be the length/volume of the region and for a test, it is the usual power (1 $-$ type II error). In the case of conformal prediction, we discussed the imitation of the optimal volume prediction region but it should be stressed that, in general, optimality is hard to attain in finite samples in a distribution-free way because {\clr it requires the transformation used in practice to be the optimal transformation and in general, one can only consistently estimate that optimal transformation under ``smoothness'' assumptions. See~\cite{MR3174619},~\cite{gyofi2020nearest}, and~\cite{yang2021finite} for some optimality results.}

In the case of testing, the optimal transformation for the equality of distribution testing would require estimation of the optimal distribution separating transformation. For instance, suppose $P_X$ and $P_Y$ are two distributions on $\mathbb{R}^d$ and in truth, they differ in their distributions only in the first coordinate. Then the optimal transformation to use is $x\in\mathbb{R}^d\mapsto x_1\in\mathbb{R}$. Among all the tests of the form suggested in Pseudocode 3, the optimal transformation should converge to $x\mapsto x_1$ asymptotically for optimality in this class of tests.  As with conformal prediction, this can be hard because of distributional assumptions and the curse of dimensionality in estimating the optimal transformation and also because the full-data transformation $\widehat{f}_{n+m}$ is not allowed to use the true $X$, $Y$ labels of $Z_1, \ldots, Z_{n+m}$ which makes it an unsupervised problem. The second issue can be alleviated by sample splitting.


In light of the discussion here, we now briefly mention a few open questions. Firstly, regarding conformal prediction, we have focused on the split conformal method for computational efficiency. This method uses one part of the data for training and the other part for calibrating the prediction region. {\clr As mentioned, it has been argued in the literature that the split conformal method could incur statistical inefficiency due to this splitting.} Because prediction regions (unlike confidence regions) do not shrink to a singleton, it is not clear how to characterize this statistical inefficiency. The results of~\cite{lei2018distribution} and~\cite{sesia2019comparison} already prove that, under certain assumptions, the split conformal regions can ``converge'' to the optimal prediction region. In this sense, asymptotic volume optimality holds in general but to understand the sub-optimality stemming from splitting, we need refined results. We believe it to be an open question on how these refined results look. Secondly, related to rank tests, we introduced sample splitting as a way of avoiding restrictions on the data-driven transformations. There is, however, a trade-off in that sample splitting tests are only based on a fraction of the total sample size and hence can also sacrifice power. It would be interesting to understand if there is a way to {\clr improve power} and make use of data more cleverly; this {\clr could be done based in p-value combination techniques~\citep{vovkcombining} or the leave-one-out analogues~\citep{barber2019predictive}. Study of power gains of such procedures requires further exploration.}
Furthermore, it would be interesting to study the (asymptotic) power properties of the tests discussed in Section~\ref{SEC:RANK-TESTS} when the data-driven transformations are assumed to be consistent (in suitable metric) to their targets.
\section*{Acknowledgments}
The author thanks Richard Berk, Andreas Buja, Rohit Patra for reading the earlier versions of the article and providing constructive comments that led to an improvement in the exposition. The author is also grateful to the reviewers, the associate editor, and the editor for their comments that led to the improved presentation.
\bibliography{AssumpLean}
\bibliographystyle{apalike}
\newpage
\appendix
\section{Auxiliary Results}\label{appsec:auxiliary-results}
The following two results follow from Theorem~\ref{thm:transformations-exchangeability} and will play an important role for both conformal prediction and rank tests. 
\begin{prop}\label{prop:split-exchangeable}
Suppose $Z_1, \ldots, Z_n, Z_{n+1}$ are exchangeable random variables. If $1 \le n_1 \le n$, and $\widehat{T}_{n_1} := g(Z_1, \ldots, Z_{n_1})$ is any statistic computed based only on $Z_1, \ldots, Z_{n_1}$, then for any function $\widehat{f}_{n_1}$ depending arbitrarily on $Z_1, \ldots, Z_{n_1}$,
\[
\widehat{f}_{n_1}(Z_{n_1+1}, \widehat{T}_{n_1}),\,\widehat{f}_{n_1}(Z_{n_1 + 2}, \widehat{T}_{n_1}),\,\ldots,\, \widehat{f}_{n_1}(Z_{n}, \widehat{T}_{n_1}),\, \widehat{f}_{n_1}(Z_{n+1}, \widehat{T}_{n_1}),
\]
are exchangeable random variables.
\end{prop}
\begin{proof}
Define the function
\[
G\begin{pmatrix}Z_1\\Z_2\\\vdots\\ Z_{n+1}\end{pmatrix} ~:=~ \begin{pmatrix}\widehat{f}_{n_1}(Z_{n_1+1}, \widehat{T}_{n_1})\\ \vdots \\ \widehat{f}_{n_1}(Z_{n+1}, \widehat{T}_{n_1})\end{pmatrix}.
\]
For any permutation $\pi:\{n_1+1,\ldots,n+1\}\to\{n_1+1, \ldots, n+1\}$,
\[
\pi G\begin{pmatrix}Z_1\\Z_2\\\vdots\\ Z_{n+1}\end{pmatrix} = \begin{pmatrix}\widehat{f}_{n_1}(Z_{\pi(n_1+1)}, \widehat{T}_{n_1})\\\cdots\\ \widehat{f}_{n_1}(Z_{\pi(n+1)}, \widehat{T}_{n_1})\end{pmatrix} = G\begin{pmatrix}Z_1\\\vdots\\Z_{n_1}\\\pi\begin{pmatrix}Z_{n_1+1}\\\vdots\\Z_{n+1}\end{pmatrix}\end{pmatrix} \overset{(a)}{=} G\begin{pmatrix}\pi_1\begin{pmatrix}Z_1\\Z_2\\\vdots\\Z_{n+1}\end{pmatrix}\end{pmatrix},
\]
for a permutation $\pi_1:[n+1]\to[n+1]$, where
\[
\pi_1(i) = i\quad\mbox{for all}\quad 1\le i\le n_1\quad\mbox{and}\quad \pi_i(i) = \pi(i)\quad\mbox{for}\quad i > n_1.
\] 
Equality (a) above follows from the fact that $\widehat{f}_{n_1}$ depends only on $Z_1, \ldots, Z_{n_1}$. Hence Theorem~\ref{thm:transformations-exchangeability} implies the result.
\end{proof}
\begin{prop}\label{prop:permutation-invariant-exchangeable}
Suppose $Z_1, \ldots, Z_n$ are exchangeable, and $\widehat{f}$ is a function depending on $Z_1, \ldots, Z_n$ permutation invariantly.\footnote{This means that the algorithm outputting $\widehat{f}$ does not use the indexing of $Z_1, \ldots, Z_{n}$.} Then $\widehat{f}(Z_1)$, $\ldots$, $\widehat{f}(Z_n)$ are exchangeable. 
\end{prop}
This result is same as Proposition 2.1 of~\cite{vovk2019testing}.
\begin{proof}[Proof of Proposition~\ref{prop:permutation-invariant-exchangeable}]
Like in Proposition~\ref{prop:split-exchangeable}, we apply Theorem~\ref{thm:transformations-exchangeability}. Define the function
\[
G\begin{pmatrix}Z_1\\\vdots\\Z_n\end{pmatrix} ~:=~ \begin{pmatrix}\widehat{f}(Z_1)\\\vdots\\\widehat{f}(Z_n)\end{pmatrix}.
\]
For any permutation $\pi:[n]\to[n]$,
\[
\pi G\begin{pmatrix}Z_1\\\vdots\\Z_n\end{pmatrix} = \begin{pmatrix}\widehat{f}(Z_{\pi(1)})\\\vdots\\\widehat{f}(Z_{\pi(n)})\end{pmatrix} = G\begin{pmatrix}Z_{\pi(1)}\\\vdots\\Z_{\pi(n)}\end{pmatrix}.
\]
The first equality above holds because $\widehat{f}$ depends on $Z_1, \ldots, Z_n$ permutation invariantly. The result now follows from Theorem~\ref{thm:transformations-exchangeability}.
\end{proof}
\section{Proofs of Results in Sections~\ref{SEC:EXCHANGEABILITY-IMPLICATIONS} and~\ref{SEC:CONFORMAL-PREDICTION}}\label{appsec:EXCHANGEABILITY-IMPLICATIONS}
\subsection{Proof of Theorem~\ref{thm:rank-distribution}}
Fix $\xi > 0$ and set $Z_i = W_i + \xi U_i$. Because $U_1, \ldots, U_n$ continuously distributed and are distinct with probability 1, we have that the events
\[
\{Z_{\pi(1)} \le Z_{\pi(2)} \le \cdots \le Z_{\pi(n)}\},
\]
over all permutations $\pi:[n]\to[n]$, are disjoint and further, one of them has to have occurred. Hence
\begin{equation}\label{eq:sum-to-one}
\sum_{\pi:[n]\to[n]}\mathbb{P}\left(Z_{\pi(1)} \le \cdots \le Z_{\pi(n)}\right) = 1.
\end{equation}
Because $U_1, \ldots, U_n$ are iid and $W_1, \ldots, W_n$ are exchangeable, $Z_1, \ldots, Z_n$ are exchangeable. This, by definition, implies that
\[
(Z_1, \ldots, Z_n) ~\sim~ (Z_{\pi(1)}, \ldots, Z_{\pi(n)}),
\]
for any permutation $\pi:[n]\to[n]$ and taking $A := \{(x_1, \ldots, x_n):\,x_1 \le \ldots \le x_n\}$ yields
\begin{align*}
\mathbb{P}(Z_1 \le \ldots \le Z_n) &= \mathbb{P}((Z_1, \ldots, Z_n)\in A)\\
&= \mathbb{P}((Z_{\pi(1)}, \ldots, Z_{\pi(n)}) \in A)\\
&= \mathbb{P}(Z_{\pi(1)} \le \ldots \le Z_{\pi(n)}).
\end{align*}
This combined with~\eqref{eq:sum-to-one} proves that for every permutation $\pi$,
\[
\mathbb{P}\left(Z_{\pi(1)} \le \cdots \le Z_{\pi(n)}\right) ~=~ \frac{1}{n!}.
\]
This proves the result because $\{Z_{\pi(1)} \le \ldots \le Z_{\pi(n)}\}$ is equivalent to the event that $(\mathrm{rank}(Z_i):i\in[n])$ is a particular permutation of $[n]$.
\subsection{Proof of Corollary~\ref{cor:for-conformal-prediction}}
Because $\mathrm{rank}(\cdot;\,\cdot)$ takes values in $\{1, 2, \ldots, n\}$, we get that
\begin{align*}
\mathbb{P}\bigg(\mathrm{rank}(W_n;\,\{W_1, \ldots, W_n\}) \le t\bigg) &= \mathbb{P}\bigg(\mathrm{rank}(W_n;\,\{W_1, \ldots, W_n\}) \le \lfloor t\rfloor\bigg)\\
&= \sum_{i=1}^{\lfloor t\rfloor} \mathbb{P}\bigg(\mathrm{rank}(W_n;\,\{W_1, \ldots, W_n\}) = i\bigg)\\ 
&= \sum_{i=1}^{\lfloor t\rfloor}\frac{(n-1)!}{n!},
\end{align*}
which proves the result.
\subsection{Proof of Theorem~\ref{thm:transformations-exchangeability}}\label{appsec:commenges}
Fix the transformation $G:\mathcal{W}^n\to(\mathcal{W}')^m$ and a vector $W$ of exchangeable random variables. Suppose for each permutation $\pi_1:[m]\to[m]$, there exists a permutation $\pi_2:[n]\to[n]$ such that 
\begin{equation}\label{eq:commenges-condition}
\pi_1G(w) = G(\pi_2w),\quad\mbox{for all}\quad w\in\mathcal{W}^n.
\end{equation}
Then for any permutation $\pi_1:[m]\to[m]$, the distribution of $\pi_1G(W)$ is same as the distribution of $G(\pi_2W)$. Because $W$ is a vector of exchangeable random variables, $\pi_2W$ has the same distribution as $W$. Therefore, $G(\pi_2W)$ (and hence, $\pi_1G(W)$) has the same distribution as $G(W)$. This completes the proof of exchangeability of the vector $G(W)$ of random variables. Hence the transformation $G$ preserves exchangeability.

To prove the second part, suppose $G:\mathcal{W}^n\to(\mathcal{W}')^m$ is a transformation preserving exchangeability. This means that whenever $W\in\mathcal{W}^n$ is a vector of exchangeable random variables, $G(W)$ is also a vector of exchangeable random variables in $(\mathcal{W}')^m$. Fix a (non-random) vector $w\in\mathcal{W}^n$ and define a random vector $W\in\mathcal{W}^n$ via the distribution
\[
\mathbb{P}(W = \pi_2'w) = \frac{1}{n!},\quad\mbox{for all permutations}\quad \pi_2':[n]\to[n]. 
\]
The idea for this distribution comes from the proof of Theorem 4 in~\cite{dean1990linear}.
Random vector $W$ has a uniform distribution on the set of all permutations of $w\in\mathcal{W}^n$. It is easy to verify that $W$ is a vector of exchangeable random variables. Because, by assumption, $G:\mathcal{W}^n\to(\mathcal{W}')^m$ is an exchangeability preserving transformation, we get that $G(W)$ is also a vector of exchangeable random variables. In particular, the support of $\pi_1G(W)$ remains constant over all permutations $\pi_1:[m]\to[m]$. Note that the support of $G(W)$ is
\[
\{G(\pi_2'w):\, \pi_2'\mbox{ a permutation on }[n]\}.
\]
Therefore, for all permutations $\pi_1:[m]\to[m]$,
\begin{equation}\label{eq:equivalence-of-support}
\begin{split}
&\pi_1\{G(\pi_2'w):\,\pi_2'\mbox{ a permutation on }[n]\}\\ 
~&\quad=~ \{G(\pi_2w):\,\pi_2\mbox{ a permutation on }[n]\}.
\end{split}
\end{equation}
Note that $\pi_1G(w)$ is an element in the set on the left hand side of~\eqref{eq:equivalence-of-support}. This implies that $\pi_1G(w)$ is equal to $G(\pi_2w)$ for some permutation $\pi_2:[n]\to[n]$. In other words, for all $w\in\mathcal{W}^n$ and all permutations $\pi_1:[m]\to[n]$, there exists a permutation $\pi_2:[n]\to[n]$ (possibly depending on $w$) such that $\pi_1G(w)=G(\pi_2w)$. The final part about necessary and sufficient condition for linear transformations follows from Theorem 4 of~\cite{dean1990linear}. This completes the proof.
\subsection{Proof of Theorem~\ref{thm:main-split-conformal}}
Proposition~\ref{prop:split-exchangeable} proves that $\widehat{f}_{\mathcal{T}}(Z_{n_1+1}), \ldots, \widehat{f}_{\mathcal{T}}(Z_{n+1})$ are exchangeable. The result now follows from Proposition~\ref{prop:conformal-prediction-main}.
\end{document}